\documentclass[a4paper,anonymous,USenglish,11pt]{article}

\usepackage{fullpage}

\usepackage{orcidlink}
\usepackage{amsfonts}
\usepackage{graphicx}
\usepackage{amsthm}
\usepackage{mathrsfs}

\usepackage{hyperref}
\usepackage{amsmath}

\usepackage{relsize}
\usepackage{caption}
\usepackage{subcaption}
\usepackage{verbatim}
\usepackage{multirow}
\usepackage[boxruled, ruled, vlined]{algorithm2e}

\usepackage{thm-restate}

\newtheorem{theorem}{Theorem}

\newtheorem{lemma}{Lemma}
\newtheorem{observation}{Observation}
\newtheorem{corollary}{Corollary}
\newtheorem{myclaim}{Claim}
\newtheorem{definition}{Definition}

\newcommand{\congest}{\ensuremath{\mathsf{CONGEST}}\xspace}
\newcommand{\local}{\ensuremath{\mathsf{LOCAL}}\xspace}
\newcommand{\clique}{\ensuremath{\mathsf{Congested\ Clique}}\xspace}
\newcommand{\bin}{\texttt{bin}\xspace}

\begin{document}

\begin{titlepage}

\author{Keren Censor-Hillel\orcidlink{0000-0003-4395-5205} \and
Majd Khoury\orcidlink{0009-0003-5125-975X}}

\title{On Distributed Computation of the Minimum Triangle Edge Transversal\thanks{Department of Computer Science, Technion, \{ckeren,majd{-}khoury\}@cs.technion.ac.il}.}


\maketitle
\thispagestyle{empty}

\begin{abstract}
The distance of a graph from being triangle-free is a fundamental graph parameter, counting the number of edges that need to be removed from a graph in order for it to become triangle-free. Its corresponding computational problem is the classic \emph{minimum triangle edge transversal} problem, and its normalized value is the baseline for triangle-freeness testing algorithms. While triangle-freeness testing has been successfully studied in the distributed setting, computing the distance itself in a distributed setting is unknown, to the best of our knowledge, despite being well-studied in the centralized setting.

~\\This work addresses the computation of the minimum triangle edge transversal in distributed networks. We show with a simple warm-up construction that this is a global task, requiring $\Omega(D)$ rounds even in the \local model with unbounded messages, where $D$ is the diameter of the network. However, we show that approximating this value can be done much faster. A $(1+\epsilon)$-approximation can be obtained in $\text{poly}\log{n}$ rounds, where $n$ is the size of the network graph. Moreover, faster approximations can be obtained, at the cost of increasing the approximation factor to roughly 3, by a reduction to the minimum hypergraph vertex cover problem. With a time overhead of the maximum degree $\Delta$, this can also be applied to the $\mathsf{CONGEST}$ model, in which messages are bounded. 

~\\Our key technical contribution is proving that computing an exact solution is ``as hard as it gets'' in $\mathsf{CONGEST}$, requiring a near-quadratic number of rounds. Because this problem is an \emph{edge selection} problem, as opposed to previous lower bounds that were for \emph{node selection} problems, major challenges arise in constructing the lower bound, requiring us to develop novel ingredients.

\end{abstract}

\end{titlepage}

\section{Introduction} 
Subgraph-freeness, and specifically triangle-freeness, are cornerstone graph properties that have been widely studied in many computational settings. Moreover, in distributed computing, triangle-free networks are known to exhibit faster algorithms for finding large cuts~\cite{HirvonenRSS17} and girth~\cite{CensorFGLLO21}, and for coloring~\cite{PettieS13}. 

For a graph that is not triangle-free, a central computational question is to compute how far it is from being so. The distance of a graph from satisfying a property has been extensively studied in the area of property testing, which asks whether the removal of any $\epsilon$-fraction of edges from the graph makes it satisfy the property~\cite{GoldreichGR98}. In particular, distributed triangle-freeness testing has been studied in~\cite{CFSV19} and in~\cite{EvenFFGLMMOORT17,FraigniaudO17}, where the latter two showed solutions within $O(1/\epsilon)$ rounds of communication in the \congest model~\cite{PelegBook}, where $n$ nodes of a network communicate in synchronous rounds with $O(\log{n})$-bit messages to their neighbors in the underlying graph.

A fundamental computational question is what is the \emph{exact} number of edges that need to be removed from a graph to make it triangle-free. This is called the \emph{minimum triangle edge transversal (MTET) problem}, and has received attention in centralized settings~\cite{ErdosGT96,KortsarzLN10,Krivelevich95} but its complexity in distributed networks has not been studied.

In this paper, we study MTET in the aforementioned classic distributed setting of \congest, as well as in the \local model~\cite{Linial92}, which is similar but without the restriction on the message size. In a distributed solution for triangle edge transversal, each node of the graph marks some of its adjacent edges, such that the set of marked edges cover all triangles in the graph, where we say that an edge covers a triangle if it is a part of that triangle. A distributed solution for MTET thus outputs a cover of minimum size.

\vspace{-0.3cm}
\subsection{Our Contribution}
As a warm-up for studying the problem, we first show that the minimum triangle edge transversal problem is \emph{global}, in the sense that solving it necessitates communication throughout the network. This translates to a tight complexity $\Theta(D)$ in the \local model, where $D$ is the diameter of the graph, as can be seen by a simple graph construction.
\begin{restatable}{theorem}{trianglesLOCALexact}
	\label{thm:trianglesLOCALexact}
	Any distributed algorithm in the \local model for computing a minimum triangle edge transversal requires $\Omega(D)$ rounds.
\end{restatable}

Despite the problem being a global one, we show that approximate solutions can be obtained fast. Using the ball-carving technique of \cite{GhaffariKM17}, we provide a $(1+\epsilon)$-approximation algorithm in the \local model.

\begin{restatable}{theorem}{trianglesLOCALonePlusEps}
	\label{theorem:OnePlusEpsInLOCAL}
	There is a distributed algorithm in the \local model that obtains a $(1+\epsilon)$-approximation to the minimum triangle edge transversal in $poly(\log{n},1/\epsilon)$ rounds.
\end{restatable}

Moreover, we show that we can approximate the task even faster, at the expense of a slightly larger approximation factor, as follows. 

\begin{restatable}{theorem}{trianglesLOCALthreeAndthreePlusEps}
	\label{theorem:threeInLOCAL}
	There are distributed algorithms in the \local model that obtain a $3$-approximation and a $(3+\epsilon)$-approximation to the minimum triangle edge transversal in $O(\log n)$ and $O(\frac{\log n}{\log\log n})$ rounds, respectively.
	The same approximation guarantees can be obtained in the \congest model with a multiplicative overhead of $\Delta$ in the round complexity.
\end{restatable}

	To obtain the above, we prove a reduction from  minimum triangle edge transversal to the \emph{minimum hypergraph vertex cover} problem (MHVC), in which one needs to find the smallest set of vertices that cover all hyperedges. The high-level outline of the reduction is that we consider each edge in the input graph $G$ as a node in a 3-uniform hypergraph $H_G$ and each triangle in $G$ is a hyperedge in it. To allow the reduction to go through, we show that we can simulate $H_G$ over $G$ without a significant overhead. Finally, we plug the MHVC algorithm of \cite{BenBasatEKS19} into our reduction to obtain the stated round complexity. 
	We note that all of our approximation algorithms apply also to the weighted version of the minimum triangle edge transversal.
	
	Finally, our key technical contribution is in showing that finding an exact solution for this problem is essentially ``as hard as it gets'' in the \congest model, requiring a near-quadratic number of rounds.

	\begin{restatable}{theorem}{trianglesCONGESTlower}
		\label{theorem:MTET_lower_bound_congest}
		Any distributed algorithm in the \congest model for computing a minimum triangle edge transversal or deciding whether there is a triangle edge transversal of size $M$ requires $\Omega(n^2/(\log^2 n))$ rounds.
	\end{restatable}
	
We emphasize that the fact that MTET is NP-hard is insufficient for deducing that it is hard for \congest. As explained in~\cite{AbboudCKP21}, there are NP-hard problems that are solvable in $o(n^2)$ rounds in \congest (and there are problems in P that require $\tilde{\Omega}(n^2)$ rounds).
To prove Theorem~\ref{theorem:MTET_lower_bound_congest}, we follow the successful approach of reducing the 2-party set disjointness problem to the distributed minimum triangle edge transversal problem. We stress that our construction requires novel insights, as we overview in what follows.

\vspace{-0.3cm}
\subsection{Technical overview of the \congest lower bound} 
	We consider the framework for reducing a 2-party communication problem to a problem $P$ in the \congest model, as attributed first to the work of~\cite{PelegR00}. Our presentation follows the one of~\cite{AbboudCKP21}. The goal is to construct a family of graphs, who differ only based on the inputs to the two parties, and show that a solution to the graph problem $P$ determines the output of the joint function $f$ that they have to compute. Then, the parties, Alice and Bob, simulate a distributed algorithm for solving $P$ to derive their output. The round complexity of the distributed algorithm is then bounded from below by the number of bits that must be exchanged between the two players for solving $f$ over their joint inputs, divided by the number of bits that can be exchanged between them per each round of the simulation. The latter is strongly affected by the size of the \emph{cut} between the two node sets that the players are responsible for simulating. Specifically, for the lower bound to be high, we need the cut between the two sets of nodes that the players simulate to be as small as possible.

	To obtain a lower bound for MTET in the \congest model, we need to find a 2-party function $f$ and to construct a family of lower bound graphs so that the value of MTET for a graph determines the value of $f$ on the joint inputs that the graph represents. The challenges for a near-quadratic lower bound are: (i) we need the communication complexity of $f$ to be linear in the input size, (ii) we need the input size of $f$ to be quadratic in the number of nodes of the lower bound graphs, and (iii) we need the cut to have at most a polylogarithmic size.

Many of the \congest problems for which lower bounds were obtained by this framework use the well-known 2-party set disjointness problem, in which the input to each of the two players is a binary string, and the output of the joint  function is 1 if and only if there exists an index in which the bit in both strings is 1.
~\\

\textbf{The Challenges}. A crucial difference between MTET and all of these problems is that \emph{MTET is an edge selection problem, rather than a node selection one}. This characteristic appears to make lower bounds harder to obtain, and we extract the central challenges below.
	
	To illustrate the difficulties, we consider the lower bound construction for the minimum vertex cover problem~\cite{AbboudCKP21}. Even without the need yet to delve into the specifics of the construction, a reasonable attempt for converting it to MTET would be to replace each node in the construction with an edge, and each edge with a triangle, so that edges cover triangles in the new construction, just as nodes cover edges in the one for minimum vertex cover. 
 
 The first glaring problem with this is that the original construction has $\Theta(n^2)$ edges, so the new one will have this number of nodes. This violates requirement (ii) above, and is a major issue since any lower bound for MTET that one could obtain this way would be at most $\tilde{\Omega}(n)$. 
 Yet, we stress that even this is insufficient, as attaching the edges into triangles in such a conversion imposes significant challenges. We focus below on two major ones.
 	
	~\\\underline{\textbf{Challenge 1:}} One obstacle is that an essential property of the lower bound graphs for the minimum vertex cover problem is their symmetry, as follows. There are four sets of cliques in the construction ($A_1,A_2,B_1,B_2$ in Figure \ref{fig:MVC lower bound graphs}), two for each player, and the input to each player corresponds to a bipartite graph between its two cliques, by adding an edge if and only if its corresponding index in the input is a 0 bit (Dashed edges in Figure \ref{fig:MVC lower bound graphs}. There are additional fixed nodes and edges which are not yet needed for the intuition here). An optimal solution for the fixed part of the construction does not \emph{prefer} any node of a clique over the others, in the sense that for each node in each clique, there is an optimal solution not containing it. 

\begin{figure}[!ht]
		\centering
		\includegraphics[width=6cm]{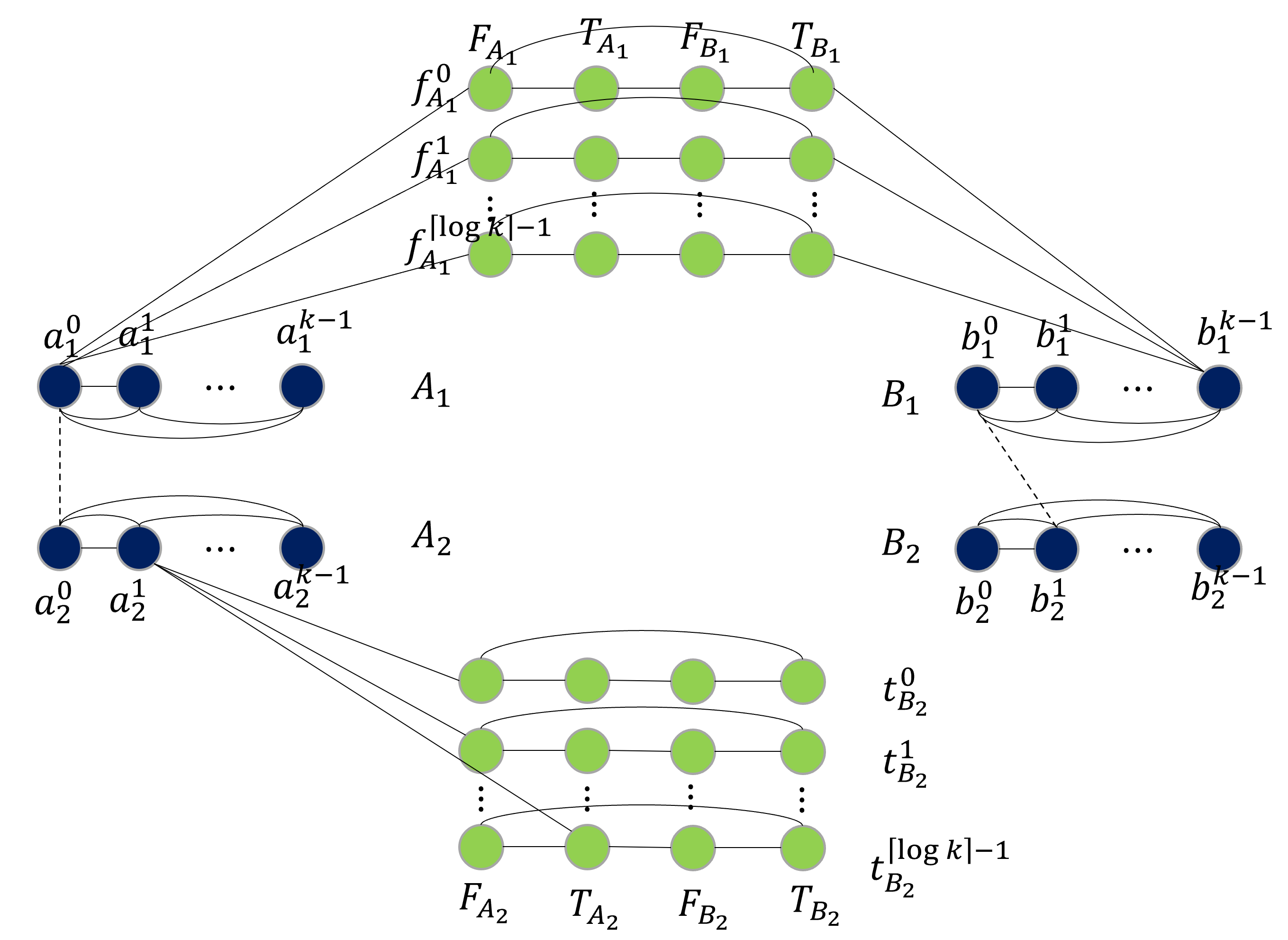}
		\caption{The family of lower bound graphs for minimum vertex cover~\cite{AbboudCKP21}.}
		\label{fig:MVC lower bound graphs}
\end{figure} 
This \emph{impartiality} feature is vital because we want to be able to choose any single one of the clique nodes to be omitted from the vertex cover, once we take into account also the need to cover the edges between the two cliques which are added according to the set disjointness input. In other words, it allows to not select a node from each of the two cliques such that this pair of nodes \emph{do not have an edge between them}, thus capturing an index in the input with a 1 bit.

 Once we try this approach for MTET, we need to switch to \emph{edge selection}. To this end, we replace each clique with an independent set of the same size and connect all nodes of the two independent sets of a certain side to a central node on that side (node $a$ for Alice and node $b$ for Bob). This implies that each input edge corresponds to a triangle between its two endpoints and the central node of its side. There will be more edges and triangles in the construction, such that we use the edges to the central node in order to cover them. Our hope is that a missing input edge (due to an input bit of 1) will allow us to avoid taking the corresponding edge from this pair of nodes to the central node into an MTET solution. However, we arrive at the following issue.
	For MTET, the only graphs that have this flavor of symmetry in optimal solutions, in the sense that they do not prefer certain edges over others, are either cliques or empty-edge graphs. We thus must attach a clique to the set of edges for which we want all but one of them to be in an optimal solution. 
 
 This obtains the property of symmetry, but doing this in a na\"ive manner prevents an optimal solution from having the ability to capture a missing input edge. To cope with this caveat, we need to configure the fixed part of the lower bound graph in a way that it prefers choosing all but one edge in certain edge sets over any other number of edges from it, or otherwise the size of the solution grows by 1 (due to the need to then take other edges in the construction that touch this set). Therefore, the solution favors choosing all but one edge from the set, \emph{but without any preference as to which edge it leaves out}. Then, for the final graph, the edge chosen to be left out corresponds to not needing to cover a missing triangle due to a missing input edge, as we wanted. Our construction that eventually overcomes the above consists of various cliques that we very carefully attach to the nodes among which we insert input edges.
	
	By using the set disjointness function, we satisfy requirement (i) above, of having the complexity of the input function be linear in its input size. Due to our above approach for the general structure of the lower bound graphs, our construction satisfies requirement (ii) of having an input size that is quadratic in the number of nodes.\footnote{It was suggested to us to replace this part of the construction with a different underlying graph with $\Theta(n)$ nodes, which is known to have $\Theta(n^2/2^{\sqrt{\log n}})$ edge-disjoint triangles~\cite{ruzsa1978triple}. Such a graph would give a smaller lower bound, but an even more severe issue is that it is not clear how to incorporate it in a construction that will allow the MTET size to correspond to the solution to the set disjointness input.}

~\\\underline{\textbf{Challenge 2:}} The above discussion provided intuition about how to construct a lower bound graph that relates its MTET \emph{for each of the two players} to their set disjointness \emph{input}. A crucial aspect that remains to be handled is how to relate \emph{the overall size} of the MTET in the entire graph to the set disjointness \emph{output}. Specifically, in our MTET solution, we need to be able to omit the same pair of edges on both sides of the construction in case the inputs to the 2-party task are not disjoint.

For the vertex cover problem, a similar property was shown using the general notion of bit-gadgets~\cite{AbboudCKP21}. Roughly speaking, each of the two main cliques on each side of the construction (to which input edges were added) had a logarithmic number of corresponding nodes to which they were attached (a clique $S\in\{A_1,A_2,B_1,B_2\}$ has the bit-gadgets $F_S$ and $T_S$, see Figure \ref{fig:MVC lower bound graphs}). These bit-gadget nodes were connected to the clique nodes according to the binary representation of their order. Each pair of bit-gadget nodes on one side are connected to the corresponding pair on the other side with a 4-cycle. This allows properly indexing the input, in the sense that if the 2-party inputs are not disjoint then an optimal solution omits the same pair of clique nodes on both sides of the construction. 

When we move from nodes that cover edges to edges that cover triangles, we must incorporate a new mechanism. To overcome this, we develop a gadget that we call a \emph{ring-of-triangles}, that replaces the simple cycles that are used for the vertex cover construction. We stress that attaching this gadget to the bit-gadget is not straightforward, and must be done in a careful manner. 
Eventually, because we make sure to still use only a logarithmic number of nodes for the cut, our construction also satisfies requirement (iii) of having a polylogarithmic cut size.

	\vspace{-0.3cm}
	\subsection{Related Work}
	Distributed triangle finding has been extensively studied, as well as finding other subgraphs (see~\cite{CH22} for a survey). In \cite{DolevLP12}, a tight $O(n^{1/3})$-round algorithm is given for listing all triangles in the \clique model~\cite{LotkerPPP05}, which resembles the \congest model but allows all-to-all communication. A matching lower bound was later established in~\cite{IzumiG17,PanduranganRS21}. Triangle listing in sparse graphs was addressed in~\cite{Censor-HillelLT20,PanduranganRS21}.
	In the \congest model, triangle listing was shown to be solvable within the same complexity in~\cite{ChangPSZ21}, by developing a fast distributed expander decomposition and using its clusters for fast triangle listing. Recently,~\cite{CensorHLV22} showed that this complexity can also be obtained in a deterministic manner, up to subpolynomial factors, based on the prior work of deterministic expander decomposition and routing of ~\cite{ChangS20}. In~\cite{HuangPZZ20}, a listing algorithm in terms of the maximum degree $\Delta$ is given.
	
	Triangle detection can be done in the \clique model in $O(n^{0.158})$ rounds~\cite{Censor-HillelKK19}, but no faster-than-listing solutions are known for it in \congest. In terms of lower bounds, it is only known that a single round does not suffice~\cite{AbboudCHKL}, even with randomization~\cite{Fischer+SPAA18}. Any polynomial lower bound would imply groundbreaking results in circuit complexity~\cite{EdenFFKO22}.

	Distributed testing algorithms were first studied in~\cite{BrakerskiP11}. Triangle-freeness testing in $O(1/\epsilon^2)$ rounds was given in ~\cite{CFSV19}, and improved to $O(1/\epsilon)$ rounds by~\cite{EvenFFGLMMOORT17,FraigniaudO17}.

	In the centralized setting, \cite{Krivelevich95} present an LP-based 2-approximation algorithm for MTET, while the integrality graph of the underlying LP was posed as an open problem. In~\cite{KortsarzLN10} a reduction is given from MTET to the minimum vertex cover problem for $(2-\epsilon)$-approximations, proving that approximating MTET to a factor less 2 is NP-hard. We stress that it is not clear how to make this reduction work in the \congest model.

	Near-quadratic lower bounds for \congest have been obtained before, e.g., for computing a minimum vertex cover, a minimum dominating set, a 3-coloring, and weighted max-cut~\cite{AbboudCKP21,bacrach2019hardness}. 
 
	\section{Hardness of Exact Computation of MTET}
	\label{sec:trianglesExact}
	
	In this section, we show that computing the exact MTET is a global problem for both the \local and \congest models, in the sense that it takes $\Omega(D)$ rounds in the \local model, and that in the \congest model its complexity is near-quadratic in $n$. We first provide some preliminaries which we use for proving the difficulty of solving the problem in both models.
	
	\subsection{Preliminaries}

	We proceed by showing properties of MTET in two special graphs, as follows. 
	
	\begin{definition}[A $t$-line of triangles]
        \sloppy{
		Let $G=\left(V,E\right)$ be an undirected graph with $t+2$ nodes, where $t\geq 2$ is even. Denote by $V=\left\{v_0,\dots,v_{t/2}\right\}\cup\left\{u_0,\dots,u_{t/2}\right\}$ the nodes of the graph.  Denote the edges by $E=\left\{\left\{v_i, v_{i+1}\right\}:0\le i \le t/2-1\right\} \cup \left\{\left\{u_i, u_{i+1}\right\}:0\le i \le t/2-1\right\}\cup\left\{\left\{v_i, u_i\right\}:0\le i \le t/2\right\} \cup\left\{\left\{v_i, u_{i+1}\right\}:0\le i \le t/2-1\right\}$. Then $G$ is called \emph{a $t$-line of triangles} (see Figure \ref{fig:line_of_triangles}).
        }
	\end{definition} 

 	\begin{figure}[!ht]
		\centering
		\includegraphics[width=6cm]{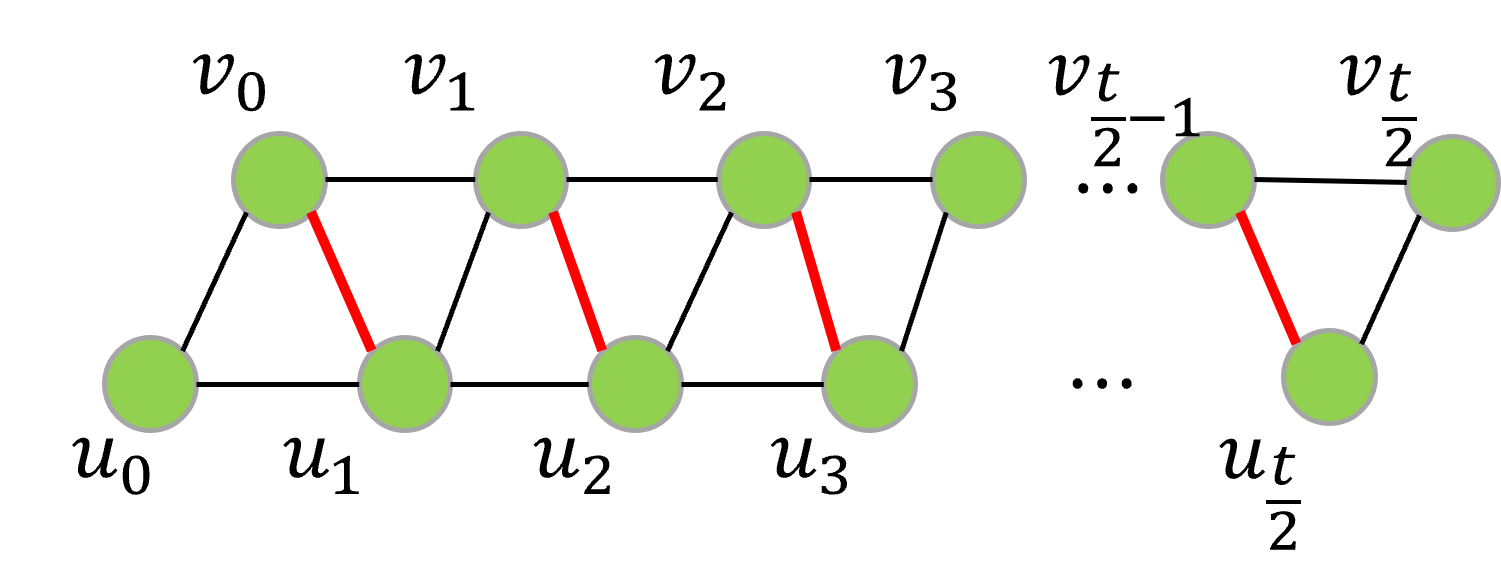}
		\caption{An illustration of a line of triangles. The bold red edges are the only optimal solution of triangle edge transversal.}
		\label{fig:line_of_triangles}
	\end{figure}

	It is immediate to see that a $t$-line of triangles  has exactly $t$ triangles. Moreover, Claim~\ref{clm:tria_line_opt} easily shows that it has a single MTET, and its size is $t/2$.

	\begin{restatable}{myclaim}{triaLineOpt}
		\label{clm:tria_line_opt}
		There exists only one MTET in a $t$-line of triangles, and its size is $t/2$.
	\end{restatable}

\begin{proof}
	\sloppy{
		By construction, every edge can cover at most 2 triangles. Thus, the minimum number of edges that is sufficient for covering $t$ triangles is at least $t/2$. The set $\left\{\left\{v_i, u_{i+1}\right\}:0\le i \le t/2-1\right\}$ is a triangle edge transversal of size $t/2$. Any other set of edges must use at least one edge that covers at most one triangle, and hence must be of size at least $t/2+1$.
	}
\end{proof}

	\begin{definition}[A $t$-ring of triangles]\label{def:ring_of_triangles}
		\sloppy{
			Let $G=\left(V,E\right)$ be an undirected graph with $t$ triangles, where $t$ is even. Then $G$ is called a \emph{$t$-ring of triangles} if the triangles can be ordered $T_0,\dots,T_{t-1}$, such that every two consecutive triangles, $T_i$ and $T_{(i+1) \mod t}, 0\le i \le t-1$, share an edge, and any other pair of triangles are edge disjoint (see Figure \ref{fig:ring_of_triangles}).
		}
	\end{definition}

 	\begin{figure}[!ht]
		\centering
		\hfill
		\begin{subfigure}[b]{0.45\textwidth}
			\centering
			\includegraphics[width=\textwidth]{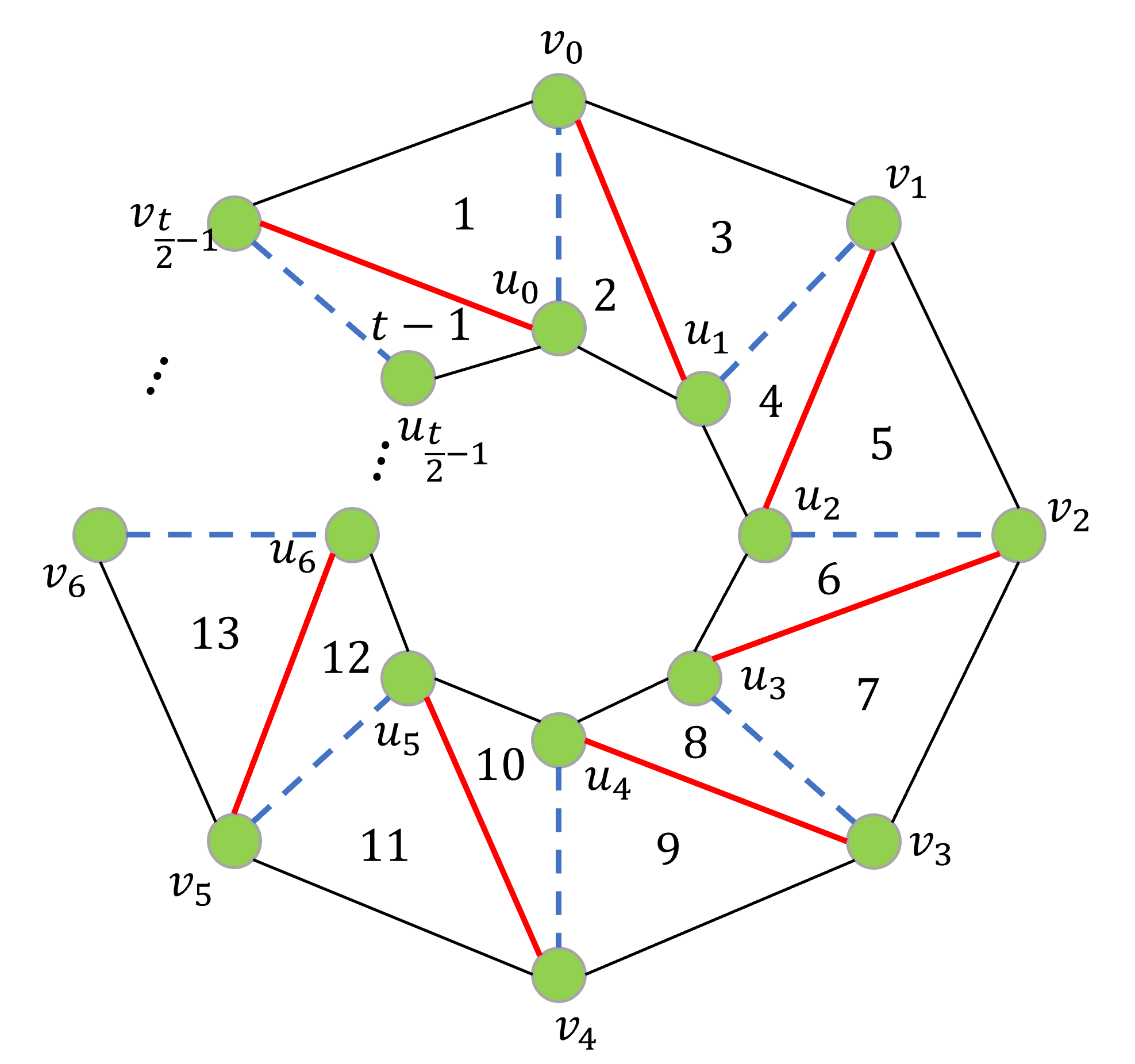}
		\end{subfigure}
		\hfill
		\begin{subfigure}[b]{0.45\textwidth}
			\centering
			\includegraphics[width=0.85\textwidth]{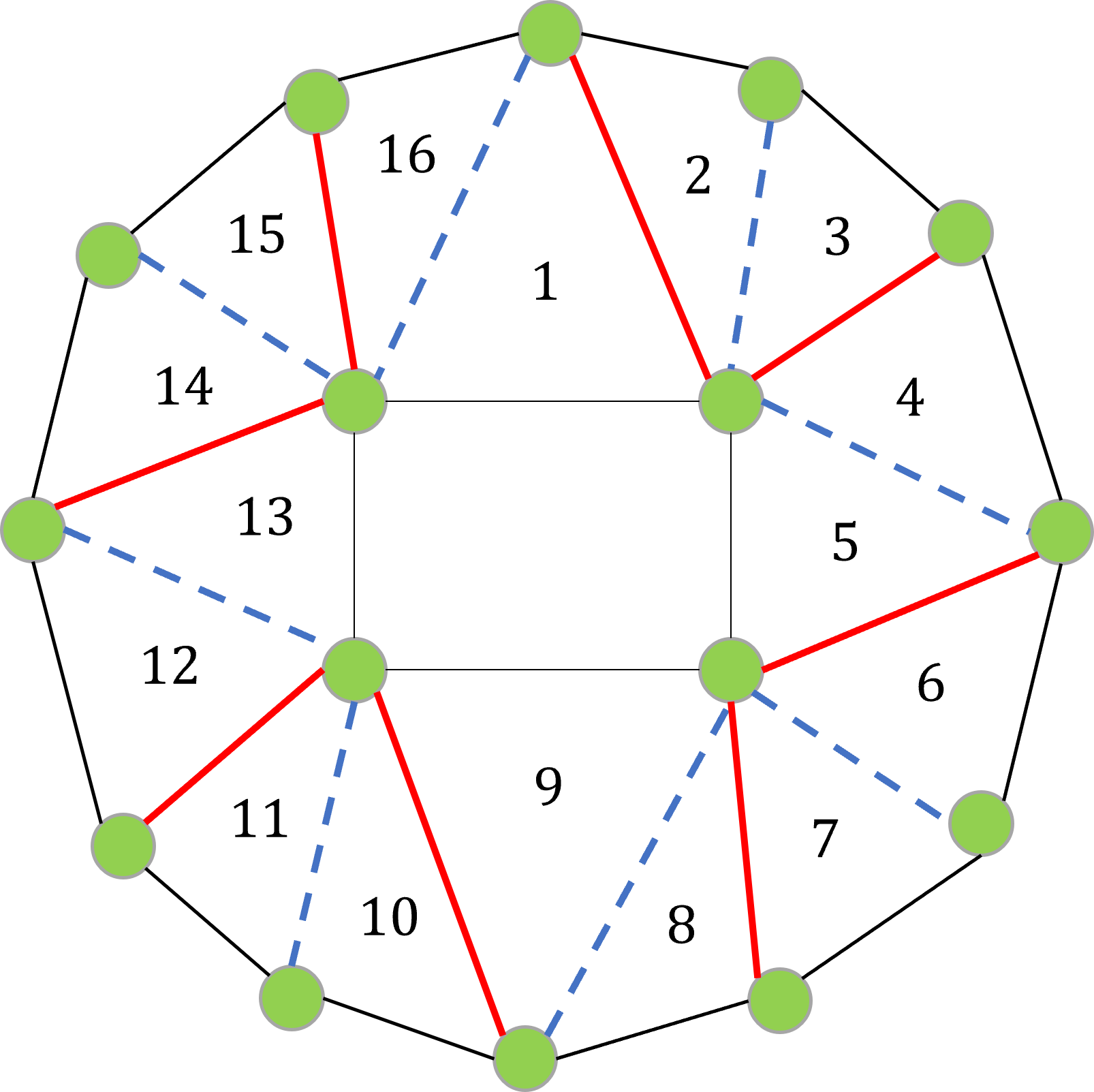}
		\end{subfigure}
		\hfill
		\caption{Two illustrations of $t$-rings of triangles. The numbers inside the triangles indicates their order. The bold red edges show an optimal solution to triangle edge transversal, and the dashed blue edges show the other optimal solution.}
		\label{fig:ring_of_triangles}
	\end{figure}

	\begin{restatable}{myclaim}{ringOfTrianglesOpt}
		\label{clm:ring_of_triangles_opt}
		There exist exactly two MTETs in a $t$-ring of triangles, and their sizes are $t/2$.
	\end{restatable}

 \begin{proof}
	\sloppy{
		By construction, each edge is incident to at most two triangles, as otherwise there are at least three triangles that share the same edge, which contradicts the definition of a $t$-ring of triangles. Therefore, every edge can cover at most 2 triangles and thus the minimum number of edges required in order to cover $t$ triangles is at least $t/2$. Let us denote the edge shared by a pair of consecutive triangles $T_i, T_{(i+1) mod t}$ by $e_i$, for $0 \le i\le t-1$ (there are exactly $t$ such edges). The set $\{e_{2i}|0\le i\le t/2-1\}$ is an optimal triangle edge transversal of size $t/2$. The set  $\{e_{2i+1}|0\le i\le t/2-1\}$ is also an optimal triangle edge transversal of size $t/2$ . Any other set of edges either has two edges that cover the same triangle or an edge that covers a single triangle, and hence must be of size at least $t/2+1$.  
	}
\end{proof}

 	\subsection{A Linear Lower Bound in the \local Model}
	\label{subsec:trianglesExactLOCAL}
	We prove that computing an exact MTET is a global problem, requiring $\Omega(D)$ rounds, where $D$ is the diameter of the graph. 
	
        \trianglesLOCALexact*

        The distinguish-ability argument is the premise of the proof, as the following: If the algorithm finishes in less than a linear number (in $D$) of rounds, then certain nodes cannot distinguish between two graphs that they should output different outputs for, and therefore output the same answer for both graphs.

	\begin{proof}
		Consider the graph $G_1=(V_1,E_1)$ which is a $t$-line of triangles for some $t\geq 2$, where $V_1=\left\{v_0,\dots,v_{t/2}\right\}\cup\left\{u_0,\dots,u_{t/2}\right\}$. And consider the graph $G_2=(V_2,E_2)$ which is a $(t+2)$-line of triangles, where $V_2=\left\{v'_0,\dots,v'_{(t+2)/2}\right\}\cup\left\{u'_0,\dots,u'_{(t+2)/2}\right\}$, which is obtained from $G_1$ by attaching two triangles to the edges on both sides (formally, the nodes $u_i$ and $v_i$ in $G_1$ correspond to the nodes $u'_{(t/2+1)-i}$ and $v'_{t/2-i}$ in $G_2$, respectively, and the nodes $u'_0, v'_{t/2+1}$ are newly added). By Claim \ref{clm:tria_line_opt}, the optimal solutions for $G_1$ and $G_2$ are $S_1 = \{\{v_i, u_{i+1}\}: 1\le i \le t/2-1\}$ and $S_2 = \{\{v'_i, u'_{i+1}\}: 1\le i \le t/2\}$, respectively. That is, the output of a node $v_i$ in $G_1$ is that the edge $\{v_i, u_{i+1}\}$ is in the optimal solution. However, in the graph $G_2$, the corresponding node $v'_{t/2-i}$ outputs that the edge $\{v'_{t/2-i}, u'_{t/2-i+1}\}$ is in the solution, which corresponds to the edge $\{v_i, u_i\}$ in $G_1$. Thus the output of $v_i$ must be different in $G_1$ and $G_2$.
		
		Note that $G_1$ has $t+2$ nodes and $G_2$ has $t+4$ nodes. Since the number of nodes are usually given to nodes, we must create similar conditions in both graphs, thus we modify the graph $G_1$ by adding two additional nodes $w_1,w_2$ both connected to $v_{t/2-1}$ and $u_{t/2}$ and call the graph $G'_1$. Notice that the solutions for MTET in $G_1$ and $G'_1$ are the same. In addition, the diameter of $G'_1$ is $D_1=\frac{t}{2}+1$ and the diameter of $G_2$ is $D_2=\frac{t}{2}+2$.
		Assume towards a contradiction that there is an algorithm that completes in $D/2-1$ rounds on every graph. Then a node in the ``middle'' of the graph cannot see the nodes on its ends. To be precise, the node $v_{\lfloor t/4\rfloor}$ in $G'_1$, which corresponds to the node $v'_{t/2-\lfloor t/4\rfloor}$ in $G_2$ sends and receives the same information in both graphs, thus outputs the same incident edge to it. This output edge is either $\{v_{\lfloor t/4\rfloor}, u_{\lfloor t/4\rfloor}\}$ in $G'_1$ and $\{v'_{t/2 - \lfloor t/4\rfloor}, u'_{(t/2 + 1) - \lfloor t/4\rfloor}\}$ in $G_2$, or $\{v_{\lfloor t/4\rfloor}, u_{\lfloor t/4\rfloor-1}\}$ in $G'_1$ and $\{v'_{t/2 - \lfloor t/4\rfloor}, u'_{t/2 - \lfloor t/4\rfloor}\}$ in $G_2$. In both cases, one of its outputs is wrong, which implies that the algorithm cannot complete in $D/2-1$ rounds.
	\end{proof}

	\subsection{A Near-Quadratic Lower Bound in the \congest Model}
	\label{subsec:trianglesExactCONGEST}
	
	In this section, we focus on proving the following theorem, which concludes a near-quadratic lower bound on the exact MTET in the \congest model.
	
	\trianglesCONGESTlower*
	
	Note that a lower bound on deciding whether there is a triangle edge transversal of size $M$ implies a lower bound on computing a minimum
	triangle edge transversal, since counting the edges of a given triangle edge transversal set can be done in $O\left(D\right)$ via standard techniques (pipelining through a BFS tree). 
	To prove the lower bound on MTET, we take the approach of reducing problems from 2-party communication along the lines pioneered in the framework of \cite{PelegR00}, and specifically, we define a family of lower bound graphs 
	\cite[Definition 1]{AbboudCKP21} for MTET. The family of lower bound graphs we define bears some similarity to the family of the lower bound class that was defined for the minimum vertex cover problem in \cite{AbboudCKP21}. However, in order to handle triangles, we need to develop much additional machinery. 
	We start by formalizing the required background.

	\paragraph{\emph{\textbf{Triangle edge transversal properties in cliques.}}}
 \sloppy{
	Let $\mathcal{ET}\left(G\right)$ be the set of MTETs of the graph $G$. Let $N\left(v\right)$ be the neighborhood of a node $v$. Let $\mu\left(G,v\right)$ denote the maximum  number of edges that touch $v$ and are in the same MTET, that is,  $\mu\left(G,v\right) = \max_{{S}\in\mathcal{ET}}\left|\left\{\left\{u,v\right\}: u\in N\left(v\right)\right\}\cap S\right|$.
 }

	\begin{definition}[Edge mapping under node mapping]
		Let $G=(V,E)$ be an undirected graph, and let $S\subseteq E$  be a subset of the edges. Let $g: V \rightarrow V$ be a bijective node mapping. Then the \emph{edge mapping} $\varphi_g$ of $S$ under $g$ is defined as:
		$\varphi_g\left(S\right) = \left\{\left\{g\left(u\right),g\left(v\right)\right\}: \left\{u,v\right\}\in S\right\}$. If $\varphi_g(S)\subseteq E$, then we call $\varphi_g$ a \emph{well-defined edge mapping}, as it can be seen as a function $\varphi_g: S \rightarrow E$.
	\end{definition}

	We start with the following simple observation, which holds since every pair of nodes in a clique has an edge.
	
	\begin{observation}
		\label{observation:well-defined-edge-mapping}
		Let $K_n=(V,E)$ be a clique with $n$ nodes, let $S\subseteq E$ be a subset of the edges, and let $g: V\rightarrow V$ be a bijective node mapping. Then $\varphi_g(S)$ and $\varphi_{g^{-1}}(S)$ are well-defined edge mappings.
	\end{observation}

	\begin{restatable}[Triangle edge transversal preserved under node mapping in cliques]{myclaim}{TETPreservation}\label{clm:TET_preservation}
		Let $K_n=\left(V,E\right)$ be a clique with $n$ nodes. Let $S$ be a triangle edge transversal of $K_n$, and let $g: V\rightarrow V$ a bijective node mapping function. Then $\varphi_g\left(S\right)$ is a triangle edge transversal.    
	\end{restatable}

      \begin{proof}
    	We show that for every triangle in the graph, there is an edge $e\in\varphi_g\left(S\right)$ covering it. Let $T=\left\{u,v,w\right\}$ be a triangle. Since $K_n$ is a clique, the set of nodes $\{g^{-1}\left(u\right),g^{-1}\left(v\right),g^{-1}\left(w\right)\}$ forms a triangle, which we denote by $T'$. Since $S$ is a triangle edge transversal, then one of the edges of $T'$ is in $S$. Without loss of generality, assume that $e'=\left\{g^{-1}\left(u\right),g^{-1}\left(v\right)\right\}$ is in $S$. Thus, applying the mapping function $\varphi_g$ on $e'$ results in an edge $e=\varphi_g\left(e'\right)$ that is in $T$. Since $\varphi_g\left(e'\right)$ is also in $\varphi_g\left(S\right)$, we have that $T$ is covered by $\varphi_g\left(S\right)$.
    \end{proof}
  We thus obtain the following.
	
	\begin{corollary}\label{col:complete_graph_edge_transversal}
		For a clique $K_n$ with $n$ nodes, the following hold:
		\begin{enumerate}
			\item For every two nodes $u,v$, it holds that $\mu\left(K_n,u\right) = \mu\left(K_n,v\right)$. Thus, we can denote this value by $\mu_n$ (where $\mu_n = \mu\left(K_n,v\right)$ for all $v$).
			
			\item For any node $v$ and $\mu_n$ nodes $u_1,u_2,\dots, u_{\mu_n}$, there exists a MTET, which we denote $S$, such that the set of edges incident to $v$ that are in $S$ is exactly $\{\left\{v,u_i\right\}, 1\le i\le \mu_n\}$.
		\end{enumerate}
	\end{corollary}

	Next, we show that the maximum number of incident edges to a node in a clique with $n$ nodes, that all are in the same MTET cannot be small, nor can it be equal to $n-1$.

	\begin{restatable}{myclaim}{maxEdgesBounds}\label{clm:max_edges_bounds}
		Let $K_n$ be a clique with $n$ nodes, then $\mu_n\ge\left(n-1\right)/6$ and $\mu_n < n-1$.
	\end{restatable}	

 	\begin{proof}	
		\sloppy{      
			Note that the number of triangles in a clique is $\binom{n}{3}$ and that each edge in a MTET of a clique covers exactly $n-2$ triangles. Therefore any MTET has at least $\frac{n(n-1)}{6}$ edges. Let $S$ be a MTET. The number of edges in $S$ that are incident to a node $v$ is $(n-1)/6$ on average, thus there is at least one node that has that many edges incident to it that are in $S$. That node sets a lower bound of $(n-1)/6$ on $\mu_n$ (by the definition of $\mu_n$).
   
			To show the upper bound on $\mu_n$, assume towards a contradiction that $\mu_n = n-1$. Then there exists a MTET $S$ and a node $v$ such that $\left\{\left\{u,v\right\}: u\in N\left(v\right)\right\} \subseteq S$. Let $w$ be a neighbor of $v$, and consider the set $S' = S \setminus \left\{\left\{v,w\right\}\right\}$. Note that the edge we removed is not in any of the triangles in the induced clique of $V \setminus \left\{v\right\}$, and thus,  $S'$ covers all the triangles in the induced clique of $V \setminus \left\{v\right\}$. All the other edges incident to $v$ are still in $S'$ and they cover all the triangles that $v$ is a part of. Thus, $S'$ is a triangle edge transversal and $|S'|=|S|-1$, in contradiction to the optimality of $S$.
		}
	\end{proof}
	
	We later incorporate the above properties of $K_n$ in our construction of the lower bound graphs.

	\textbf{2-Party Communication Complexity.}
	In the 2-party communication complexity setting~\cite{KushilevitzN97}, there are two players, Alice and Bob, who are given two inputs $x,y\in \{0,1\}^K$, respectively, and wish to evaluate a function $f:\{0,1\}^K\times\{0,1\}^K\rightarrow \{\texttt{true}, \texttt{false}\}$ on their joint input.
	The communication between Alice and Bob is carried out according to some fixed protocol $\pi$ (which only depends on the function $f$, known to both players). The protocol consists of the players sending bits to each other until the value of $f$ is determined by either of them.
	The maximum number of bits that need to be sent in order to compute $f$ over all possible inputs $x,y$, is called \emph{the communication complexity} of the protocol $\pi$, and is denoted by $CC(\pi)$. The \emph{deterministic communication complexity} for computing $f$, denoted $CC(f)$, is the minimum over $CC(\pi)$  taken over all deterministic protocols $\pi$ that compute $f$. 
	In a randomized protocol, Alice and Bob have access to (private) random strings $r_A, r_B$, respectively, of some arbitrary length, chosen independently according to some probability distribution. The \emph{randomized communication complexity} for computing $f$, denote ${CC}^R(f)$, is the minimum over $CC(\pi)$ taken over all randomized protocols $\pi$ that compute $f$ with success probability at least $2/3$.

	\textbf{Set disjointness.} The \emph{Set Disjointness} function, denoted by $f(x,y)=DISJ_K(x,y)$, whose inputs are of size $K$, returns $\texttt{false}$ if the inputs represent sets that are not disjoint, i.e, if there exists an index $i\in\{0,\dots,K-1\}$ such that $x_i=y_i=1$, and returns $\texttt{true}$ otherwise.
	The deterministic and randomized communication complexities of $DISJ_K$ are $\Omega(K)$ \cite{Bar-YossefJKS04,KalyanasundaramS92,Razborov90}.

	~\\
	\textbf{Lower Bound Graphs.}
	We quote the following formalization of the reduction of optimization problems from 2-party communication complexity to distributed problems in \congest. 
	
	\begin{definition}
		Fix an integer $K$, a function $f:\left\{0,1\right\}^K\times\left\{0,1\right\}^K\rightarrow\left\{\texttt{true},\texttt{false}\right\}$, and a predicate $P$ for graphs. The family of graphs $\{G_{x,y}=\left(V,E_{x,y}\right):x,y\in\left\{0,1\right\}^K\}$, is said to be \emph{a family of lower bound graphs w.r.t $f$ and $P$} if the following properties hold:
		\begin{enumerate}
			\item The set of nodes $V$ is the same for all graphs, and $V = V_A\cup V_B$ is a fixed partition of it;
			\item Only the existence or the weight of edges in $V_A \times V_A$ may depend on $x$;
			\item Only the existence or the weight of edges in $V_B \times V_B$ may depend on $y$;
			\item $G_{x,y}$ satisfies the predicate $P$ iff $f\left(x,y\right)=\texttt{true}$.
		\end{enumerate}
	\end{definition}
	
	\sloppy{
		The following theorem reduces 2-party communication complexity problems to \congest problems (see e.g \cite{AbboudCKP21,DasSarmaHKKNPPW11,DruckerKO13,FrischknechtHW12,HolzerP15}).
	}
	
	\begin{theorem}\label{thm:CC_reduction}
		Fix a function $f:\left\{0,1\right\}^K\times\left\{0,1\right\}^K\rightarrow\left\{\texttt{true},\texttt{false}\right\}$ and a predicate $P$. Let $\{G_{x,y}\}$ be a family  of lower bound graphs w.r.t $f$ and $P$ and denote  $C=E\left(V_A,V_B\right)$. Then any deterministic algorithm for deciding $P$ in the \congest model requires $\Omega\left(CC\left(f\right)/(|C|\log n)\right)$ rounds, and any randomized algorithm for deciding $P$ in the \congest model requires $\Omega\left(CC^R\left(f\right)/(|C|\log n)\right)$ rounds. 
	\end{theorem}
	
	\subsubsection{The Lower Bound.}

	We begin by constructing the fixed graph on which we proceed by adding edges corresponding to the input strings $x$ and $y$ of the 2-party set disjointness problem. Then we show that having a solution for MTET of size $K$ implies the disjointness of the input strings and vice versa, hence the complexity of the set disjointness problem provides a lower bound to the complexity of the MTET problem, according to Theorem \ref{thm:CC_reduction}. The construction is illustrated by Figures \ref{fig:lower bound fixed graphs}, \ref{fig:fixed graph bit gadget connections}, \ref{fig:fixed graph connectors connections} and \ref{fig:fixed graph rings connections}.
	
	~\\
	\textbf{The fixed graph construction -- the node set.} 
	\sloppy{
		Let $k$ be a natural number. The fixed graph has two center nodes $a$ and $b$, and four sets of nodes $A_i=\left\{a_i^\ell: 0\le \ell \le \mu_{k+1}\right\}, B_i=\left\{b_i^\ell: 0\le \ell \le \mu_{k+1}\right\}, i \in \left\{1,2\right\}$, which we call \emph{bit nodes}. For each set $S\in \left\{A_1,A_2,B_1,B_2\right\}$, the graph also has four additional corresponding sets of nodes, as follows. Two such sets are \emph{bit-gadgets} (as defined in~\cite{AbboudCK16}  and explained below) of size $\lceil\log\left(\mu_{k+1}+1\right)\rceil$, $T_S=\left\{t_S^\ell: 0\le \ell \le \lceil\log\left(\mu_{k+1}+1\right)\rceil-1\right\}$  and $F_S=\left\{f_S^\ell: 0\le \ell \le \lceil\log\left(\mu_{k+1}+1\right)\rceil-1\right\}$, another set is a clique of size $k$, $C_S=\left\{c_S^\ell: 0 \le \ell \le k-1\right\}$, and another set $H_S = \left\{h_S^\ell: 0 \le \ell \le \mu_{k+1}\right\}$ called \emph{connectors} is of size $\mu_{k+1}+1$ (this set connects between the bit nodes to the corresponding cliques with a line of $2$-triangles as will be explained later and illustrated in Figure \ref{fig:fixed graph connectors connections}). In addition, we add \emph{ring-auxiliary nodes} $M^\ell_i = \{m^j_{i,\ell}, 0\le j \le 13\}$, for $i\in{1,2}$ and $0\le \ell \le \lceil\log\left(\mu_{k+1}+1\right)\rceil-1$, that connect the edges $\{f^\ell_{A_i},a\},\{t^\ell_{A_i},a\},\{f^\ell_{B_i},b\},\{t^\ell_{B_i},b\}$ as part of a 20-ring of triangles, as defined in Definition  \ref{def:ring_of_triangles} (see Figure \ref{fig:fixed graph rings connections}, where the bold edges indicate the four previously mentioned edges for $M^0_1$, that is, for $i=1$ and $\ell=0$).
	}

	~\\\textbf{The fixed graph construction -- the edge set.} 
	For clarity, we divide the edge set to five sets and explain the contents of each set.
	
	\underline{Central edges:} The center node $a$ is connected to all the nodes in the bit-gadgets $F_{A_1}, T_{A_1}, F_{A_2}, T_{A_2}$, the cliques $C_{A_1}, C_{A_2}$ and the connector nodes $H_{A_1}, H_{A_2}$. The center node $b$ is connected to all the nodes in the bit-gadgets $F_{B_1}, T_{B_1}, F_{B_2}, T_{B_2}$, the cliques $C_{B_1}, C_{B_2}$, and the connector nodes $H_{B_1}, H_{B_2}$ (see Figures \ref{fig:lower bound fixed graphs} and \ref{fig:fixed graph connectors connections}).

	\underline{Clique edges:} Each of the sets $C_{A_1},C_{A_2},C_{B_1},C_{B_2}$ is a clique, i.e, is isomorphic to the graph $K_{k}$. Therefore, every two nodes in the same set $C_S$ where $S\in\{A_1,A_2,B_1,B_2\}$ are connected (see Figure \ref{fig:lower bound fixed graphs}).
	
	\underline{Bit edges:}
	\sloppy{
		The sets $T_S,F_S$ behave as bit-gadgets, meaning that the connections between the bit-gadgets and their corresponding bit nodes in $S$ match their binary representation. Formally, for $(s,i)\in\{(a,1), (a,2), (b,1), (b,2)\}$ and $0\le \ell \le \mu_{k+1}+1$, let $s^\ell_i$ be a node in the corresponding set of bit nodes (i.e, in $\{A_1,A_2,B_1,B_2\}$), and let $(\ell)_j$ be the bit value of $j$th bit in the binary representation of $\ell$. Then the \emph{corresponding binary representation nodes} of $s^\ell_i$, denoted $\bin(s^\ell_i)$, are $\bin(s^\ell_i)=\left\{f^j_S | (\ell)_j = 0\right\}\cup \left\{t^j_S | (\ell)_j = 1\right\}$ for the corresponding set $S\in\{A_1,A_2,B_1,B_2\}$ that the node $s^{\ell}_i$ is a part of. 
		Therefore, each bit node $s^\ell_i$ is connected to $\bin(s^\ell_i)$ (see Figure \ref{fig:fixed graph bit gadget connections}).
	}
	
	\underline{Connector edges:}
	For any $(s,i,S)\in\{(a,1,A_1),(a,2,A_2),(b,1,B_1),(b,2,B_2)\}$, each node  $h_S^\ell$ in $H_S$  is connected to the nodes with the corresponding index in $S$ and the clique $C_S$, i.e, to the nodes $s^\ell_i$ and $c^\ell_S$ (see Figure \ref{fig:fixed graph connectors connections}). Notice that the connections from $H_S$ to $S$ and $C_S$ are possible because 
	the size of $S$ equals the size of $H_S$, and because due to the inequality $\mu_{k+1} < k$ that is satisfied according to Claim \ref{clm:max_edges_bounds}, it holds that the number of nodes in $C_S$ is sufficiently large to connect all the nodes in $H_S$ to $C_S$.

	\sloppy{
		\underline{Ring edges:} For $i\in\{1,2\}$, the edges between $a$ and $F_{A_i}, T_{A_i}$ (that is, the edges $\{a, f^\ell_{A_i}\}, \{a, t^\ell_{A_i}\}$) participate in a 20-ring of triangles with the edges between $b$ and $F_{B_i}, T_{B_i}$ (that is, the edges $\{b, f^\ell_{B_i}\}, \{b, t^\ell_{B_i}\}$) with the ring-auxiliary nodes acting as connectors. To be specific, let $i\in\{1,2\}$ and $0\le \ell\le \lceil\log\left(\mu_{k+1}+1\right)\rceil-1$. Then the following edges are part of the set of ring edges: $\{\{a,m^j_{i,\ell}\}\mid j\in\{0,1,12\}\}$, $\{\{b,m^j_{i,\ell}\}\mid j\in\{5,6,7\}\}$, $\{\{t^\ell_{A_i},m^j_{i,\ell}\}\mid j\in\{1,2,3\}\}$, $\{\{f^\ell_{B_i},m^j_{i,\ell}\}\mid j\in\{3,4,5\}\}$, $\{\{t^\ell_{B_i},m^j_{i,\ell}\}\mid j\in\{6,7,8\}\}$, $\{\{f^\ell_{A_i},m^j_{i,\ell}\}\mid j\in\{0,12,13\}\}$, as well as the edges $\{\{m^j_{i,\ell}, m^p_{i,\ell}\}\}$ for any $(j,p)$ that are $(0,1)$, $(2,3)$, $(2,4)$, $(3,4)$, $(5,6)$, $(7,8)$, $(7,9)$, $(8,9)$, $(8,10)$, $(9,10)$, $(9,11)$, $(10,11)$, $(10,13)$, $(11,12)$, $(11,13)$, or $(12,13)$ (see Figure \ref{fig:fixed graph rings connections}). Note that the choice of having exactly 20 triangles per ring is arbitrary and any even number greater than 12 works. The important part is to have an odd number of triangles (at least 3) between the four pairs of edges $(\{f^\ell_{A_j},a\}, \{t^\ell_{A_j},a\}), (\{t^\ell_{A_j},a\}, \{f^\ell_{B_j},b\}), (\{f^\ell_{B_j},b\},\{t^\ell_{B_j},b\})$, and $(\{t^\ell_{B_j},b\}, \{f^\ell_{A_j},a\})$.
	}

	\begin{figure}[!ht]
		\centering
		\includegraphics[width=12cm]{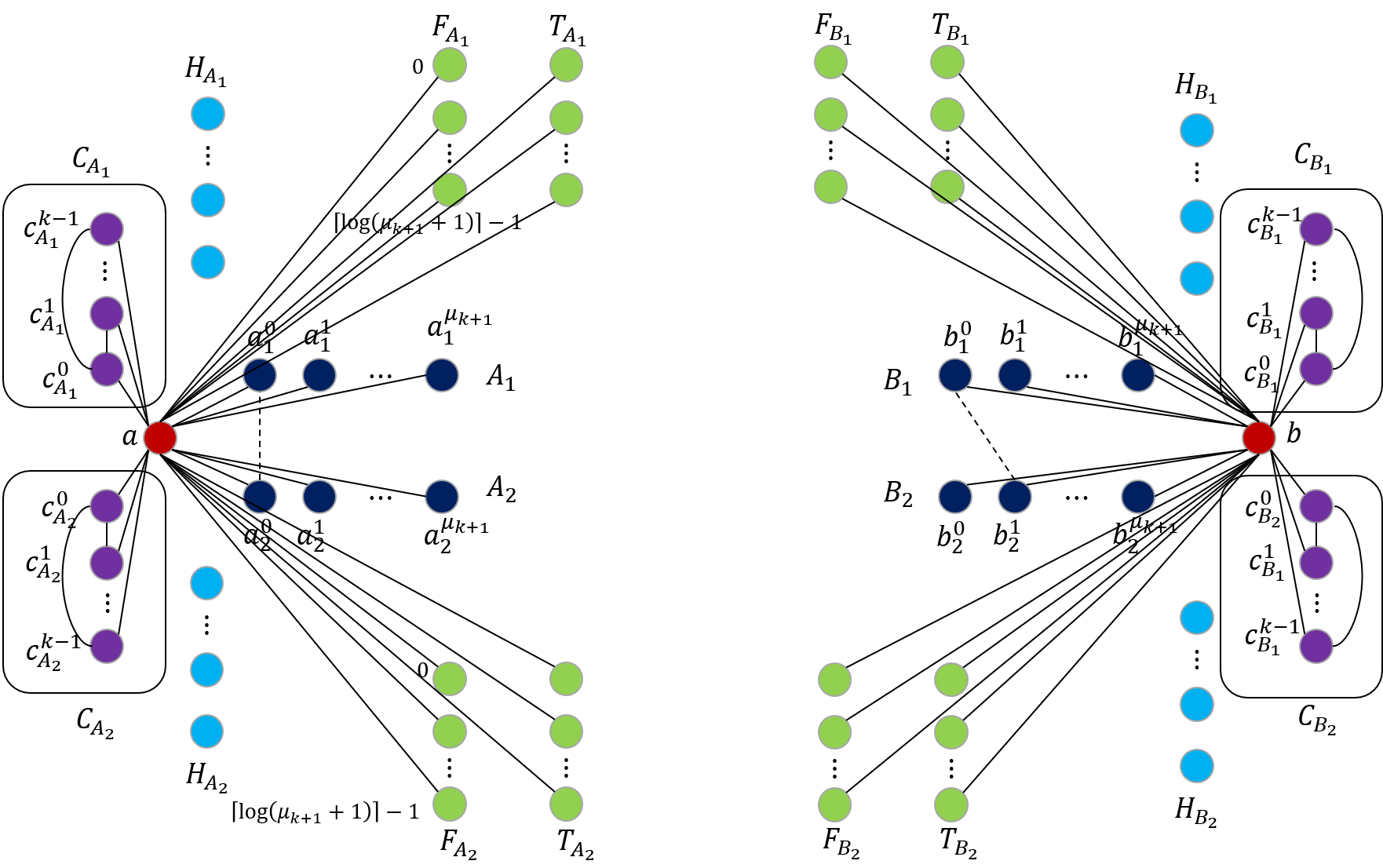}
		\caption{The basic structure of the family of lower bound graphs for deciding the size of the minimum triangle edge transversal, with many edges and nodes
			omitted for clarity. See the additional figures for more detailed illustrations.}
		\label{fig:lower bound fixed graphs}
	\end{figure}
	
	\begin{figure}[!ht]
		\centering
		\includegraphics[width=11cm]{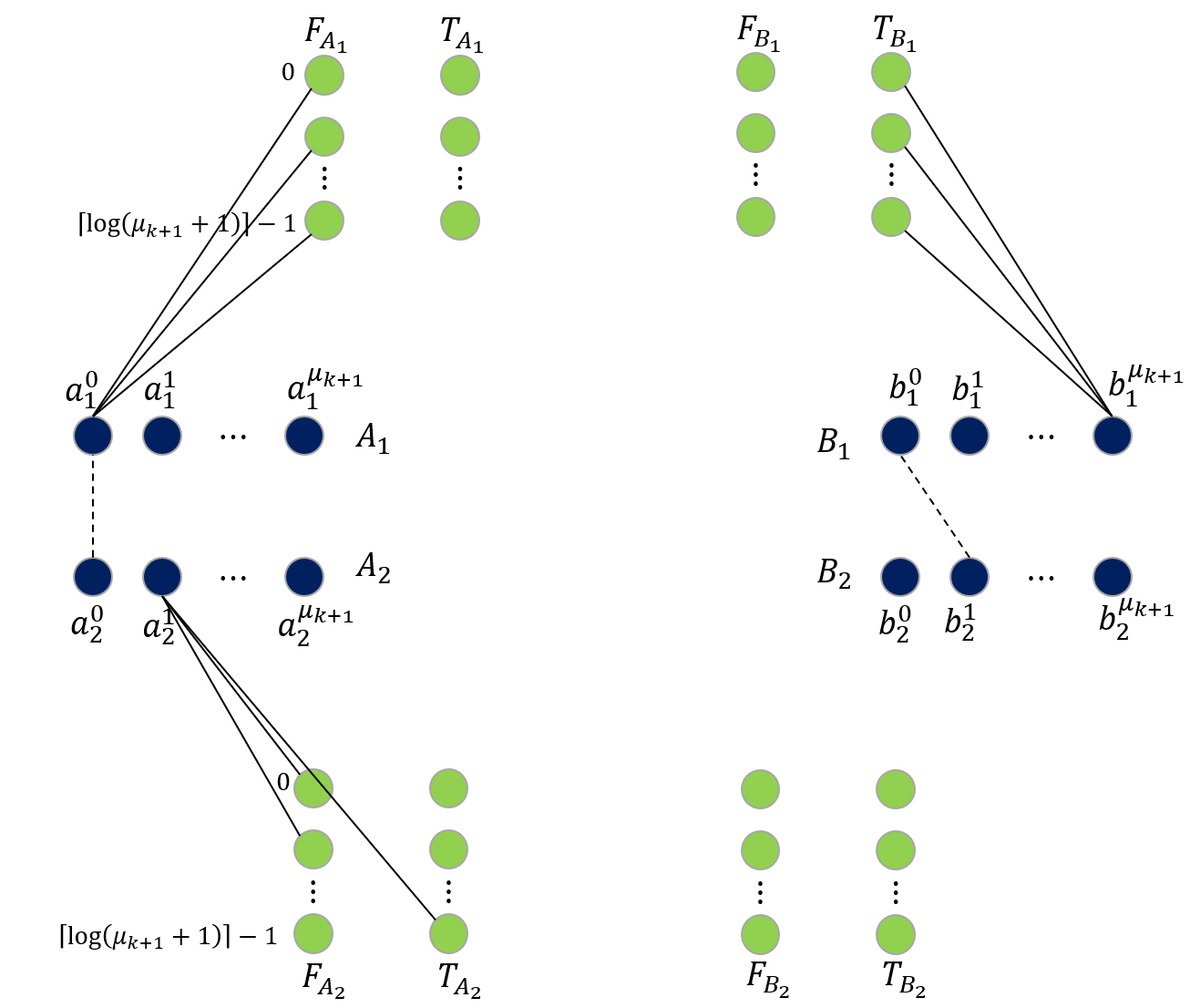}
		\caption{The bit-gadget layer of connections, each bit node is connected to the corresponding nodes in the bit gadget with respect to the binary representation of that node.}
		\label{fig:fixed graph bit gadget connections}
	\end{figure}
	
	\begin{figure}[!ht]
		\centering
		\includegraphics[width=12cm]{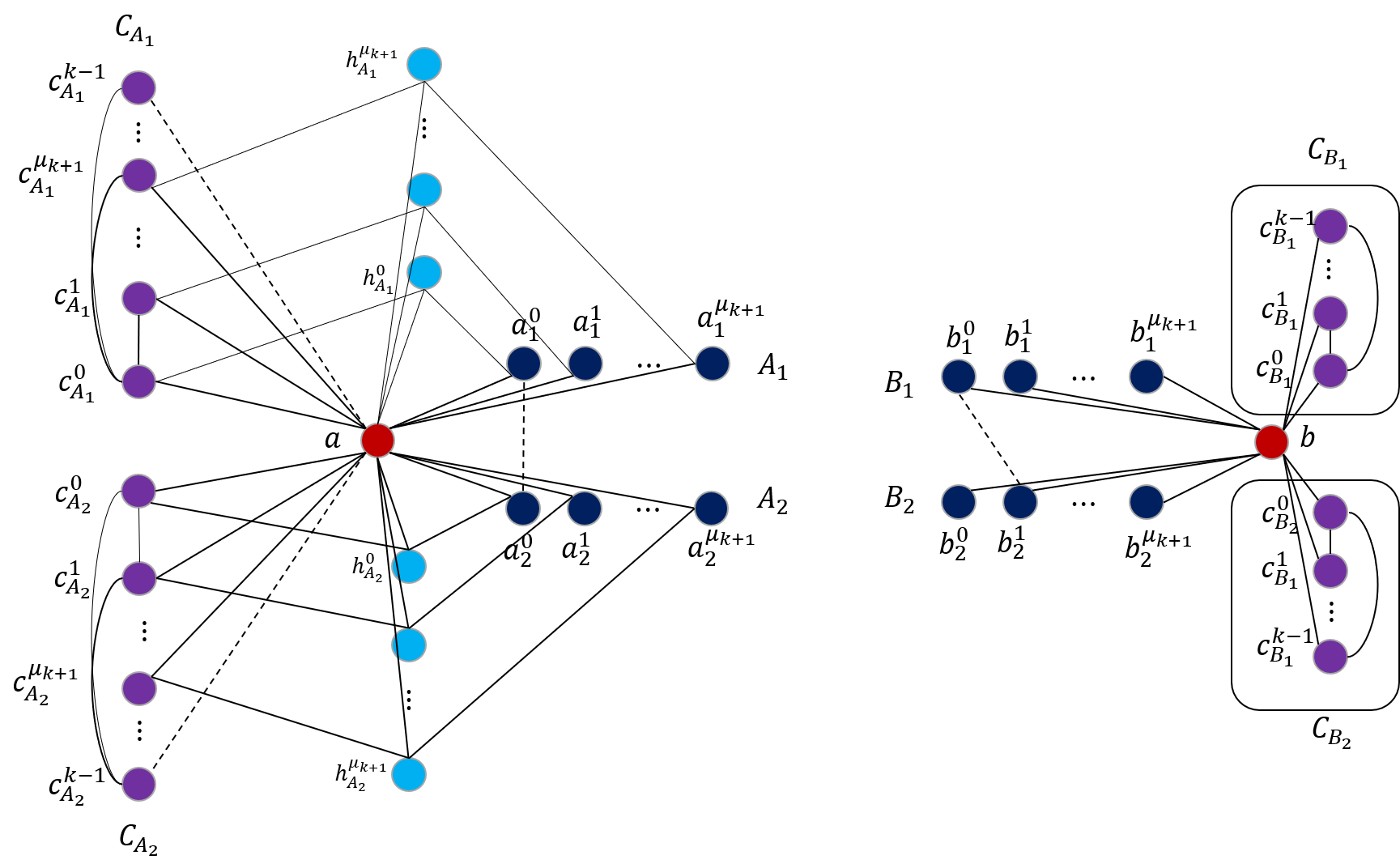}
		\caption{The clique layer of connections (Connector edges), each edge between a bit node and a center node (e.g., $\{a,a^i_1\}$) and the edge between a clique node with the same index and the center node (i.e., $\{a,c^i_{A_1}\}$)  are connected via a line of 2-triangles ($\{a,h^i_{A_1},c^i_{A_1}\}$ and $\{a,h^i_{A_1},a^i_1\}$).}
		\label{fig:fixed graph connectors connections}
	\end{figure}

	\begin{figure}[!ht]
		\centering
		\includegraphics[width=14cm]{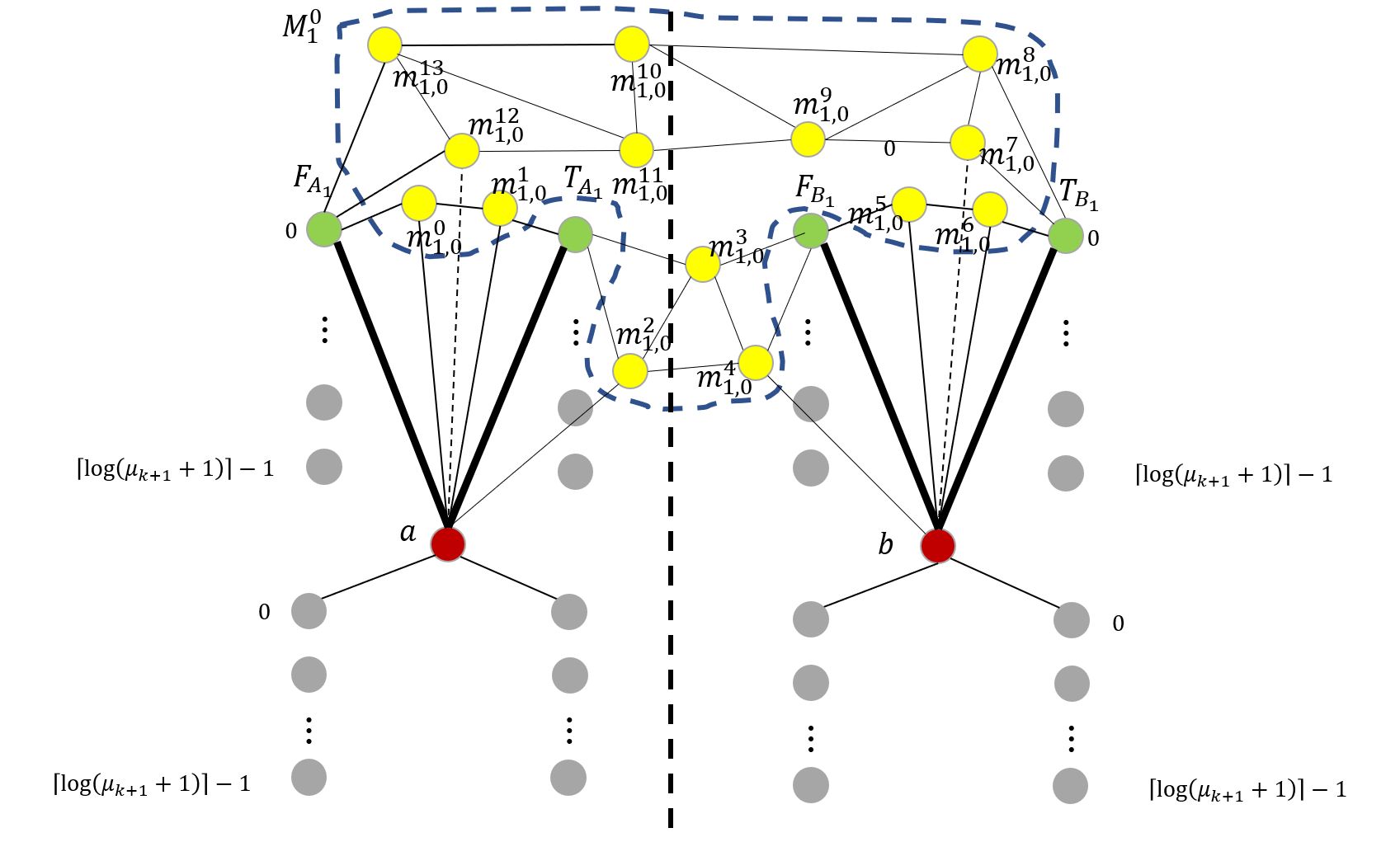}
		\caption{This figure shows the connection between the bit-gadgets. Each tuple of the four edges $S_i^j=\left(\left\{a,f^i_{A_j}\right\}, \left\{a,t^i_{A_j}\right\}, \left\{b,f^i_{B_j}\right\}, \left\{b,t^i_{B_j}\right\}\right)$ for $0 \le i \le \lceil\log\left(\mu_{k+1}+1\right)\rceil+1$ and $j\in \{1,2\}$ are connected with each other via a 20-ring of triangles, such that between each two consecutive edges of $S_i^j$ (cyclic) there is an odd number of triangles. For example, between $\{a,f^i_{A_j}\}$ and $\{a,t^i_{A_j}\}$ there are three triangles and between $\{b,t^i_{B_j}\}$ and $\{a,f^i_{A_j}\}$ there are 9 triangles.}
		\label{fig:fixed graph rings connections}
	\end{figure}

	This concludes the construction of the fixed graph. We next analyze its required properties, describe the edges that are added according to the set disjointness inputs strings $x,y$, and complete the lower bound proof.
	We denote the size of the MTET of a graph by $\tau(G)$.
	\begin{myclaim}\label{clm:transversal lower bound partition}
		Any triangle edge transversal of the fixed graph $G$ must contain $\tau\left(K_{k+1}\right)$ edges of each clique, 10 edges of each 20-ring of triangles (Figure \ref{fig:fixed graph rings connections}), and an edge from each triangle $\left\{a_i^j, h_{A_i}^j ,a\right\}$ and $\left\{b_i^j, h_{B_i}^j,b\right\}$, for $i\in \left\{1,2\right\}, 0 \le j \le \mu_{k+1}$.
	\end{myclaim}
	
	Instead of proving the claim directly, we prove a more generalized claim, which implies it.
	
	\begin{myclaim}[Sub-optimality of parts of a triangle edge transversal]\label{clm:sub_optimal_parts}
		Let $E'\subseteq E$ be a subset of edges of $G$. Let $H=\left(V,E'\right)$ be the subgraph of $G$ with only the edges from $E'$. Then every triangle edge transversal $S$ of G contains at least $\tau\left(H\right)$ edges of $E'$.
	\end{myclaim}

	\begin{proof}
		The set $S\cap E'$ is a triangle edge transversal of $H$, since all the triangles that are fully contained in $H$ can be covered using only edges of $S\cap E'$. Thus, $|S\cap E'| \ge \tau\left(H\right)$. Note that $S\cap E' \subseteq E'$ and $S\cap E' \subseteq S$, thus, $S$ contains at least $|S\cap E'| \ge \tau\left(H\right)$ edges of $E'$.
	\end{proof}
	
	Having multiple edge-disjoint subsets of the edges gives a stronger lower bound on the size of a triangle edge transversal.
	
	\begin{myclaim}\label{clm:tria_edge_lower_bound}
		Let $E_1,\dots,E_m$ be disjoint subsets of $E$, and for every $1\leq i\leq m$ let $H_i=\left(V,E_i\right)$ be the subgraph of $G$ with only the edges of $E_i$. Then every triangle edge transversal $S$ contains at least $\tau\left(H_i\right)$ edges of $E_i$ and hence $|S| \ge \sum_{i=1}^m |\tau\left(H_i\right)|$.
	\end{myclaim}
	
	\begin{proof}
		The first part of the claim follows directly from Claim \ref{clm:sub_optimal_parts}. The second part follows from the disjointness of the subsets.
	\end{proof}

	\begin{restatable}[Corollary of Claims \ref{clm:transversal lower bound partition} and \ref{clm:tria_edge_lower_bound}]{corollary}{fixedGraphMTETLowerBound}
		\label{col:fixed_graph_MTET_lower_bound}
		The size of any triangle edge transversal of the fixed graph $G$ is at least $4\tau\left(K_{k+1}\right)+4\left(\mu_{k+1}+1\right)+2\cdot10\lceil\log\left(\mu_{k+1}+1\right)\rceil$.
	\end{restatable}

 \begin{proof}
	Consider the following disjoint subsets of edges.
	\begin{itemize}
		\item The edges between each clique $C_S$ to the center node, i.e.,
		the edges between all nodes in $C_S\cup \{s\}$ for $(s,S)\in\{(a,A_1),(a,A_2),(b,B_1),(b,B_2)\}$
		\item The edges of the triangles that touch the connectors $H_S, S\in\{A_1,A_2,B_1,B_2\}$, i.e, the triangles $\{s^j_i, h^j_{S_i}, s\}$ for $i\in\{1,2\}, 0\le j\le \mu_{k+1}, (s,S)\in\{(a,A),(b,B)\}$ (Figure \ref{fig:fixed graph connectors connections})
		\item The edges of each 20-ring of triangles (Figure \ref{fig:fixed graph rings connections})
	\end{itemize}
	The size of the MTET on the subgraph with only the edges of each of the four cliques in item (1) is $\tau(K_{k+1})$. The size of the MTET on each triangle in item (2) is 1. The size of the MTET on each 20-ring of triangles is 10 according to Claim \ref{clm:ring_of_triangles_opt}. This concludes the analysis of the required size for any solution.
\end{proof}

\begin{restatable}{myclaim}{bitCorrespondence}
	\label{clm:bit_correspondence}
	If $S\subseteq E$ is a triangle edge transversal of $G$ of size $4\tau\left(K_{k+1}\right)+4\left(\mu_{k+1}+1\right)+2\cdot10\lceil\log\left(\mu_{k+1}+1\right)\rceil$, then there are two indices $i,j \in \left\{0, \dots,\mu_{k+1}\right\}$ such that the edges $\left\{a,a_1^i\right\}, \left\{a,a_2^j\right\}, \left\{b,b_1^i\right\}, \left\{b,b_2^j\right\}$ are not in $S$.
\end{restatable}

\begin{proof}
	If the size of the triangle edge transversal $S$ is $4\tau\left(K_{k+1}\right)+4\left(\mu_{k+1}+1\right)+2\cdot10\lceil\log\left(\mu_{k+1}+1\right)\rceil$, then it includes the minimum number of edges needed to cover the triangles in each of the subgraphs mentioned in the proof of Corollary \ref{col:fixed_graph_MTET_lower_bound}, and those are the only edges in $S$. Therefore, the number of edges that connects $a$ to the clique $C_{A_1}$ and are in $S$ is at most $\mu_{k+1}$ by Claim \ref{clm:max_edges_bounds}. Hence, there is at least one edge in the subset $\left\{\left\{a, c_{A_1}^\ell\right\}: 0\le \ell\le \mu_{k+1}\right\}$ that is not is $S$, as this subset is of size $\mu_{k+1}+1$, and let us denote by $i$ the index of the node in the clique that corresponds to this edge (i.e, it is the node $c_{A_1}^i$). To cover the triangle $\{c_{A_1}^i, h_{A_1}^i,a\}$, if we take $\{a, h^i_{A_1}\}$ then it increases the size of $S$, since this edge is not a part of any of the subgraphs mentioned in the proof of Corollary, thus violating the assumption. This forces the edge $\{a, h_{A_1}^i\}$ to be a part of $S$ (see Figure \ref{fig:fixed graph connectors connections}), and thus $\{a, a_1^i\}$ cannot be in $S$ in order to avoid increasing its size.
	
	\sloppy{
		The same explanation applies to the edges $\left\{\left\{a, c_{A_2}^\ell\right\}: 0\le \ell\le \mu_{k+1}\right\}$, $\left\{\left\{b, c_{B_1}^\ell\right\}: 0\le \ell\le \mu_{k+1}\right\}$ and  $\left\{\left\{b, c_{B_2}^\ell\right\}: 0\le \ell\le \mu_{k+1}\right\}$, in that exactly one edge of each subset is not part of $S$. Let us denote those edges by $\{a, c_{A_2}^j\}, \{b, c_{B_1}^{i'}\}, \{b, c_{B_2}^{j'}\}$. Since these edge are not part of $S$, the edges $\{a, h_{A_2}^i\}$, $\{a, h_{B_1}^i\}$, $\{a, h_{B_2}^i\}$ must be part of $S$, thus omitting the edges $\{a, a_2^j\}$, $\{b, b_1^{i'}\}$, $\{b, b_2^{j'}\}$.
	}
	
	Let us assume that $i\ne i'$, thus, there exists a bit in a specific index $\ell$ in their binary representation where $i$ and $i'$ differ. Since the edges $\{a, a_1^i\}, \{b,b_1^{i'}\}$ are not in $S$, then all the edges between $a$ and $\bin\left(a_1^{i}\right)$ and the edges between $b$ and $\bin\left(b_2^{i'}\right)$ are part of $S$ (recall that $\bin\left(a_1^{i}\right)$ are the nodes in the bit gadget corresponding to the binary representation of $a_1^{i}$). Let us assume without loss of generality that the $\ell$-th bits of the binary representations of $a_1^{i}$ and $b_1^{i'}$ are 1 and 0,  respectively. Note that at the index $\ell$ where $i$ and $i'$ differ in their binary representation, the edges $\{a, t_{A_1}^\ell\}, \{b, f_{B_1}^\ell\}$ were added to $S$, however, there is no optimal solution for the $\ell$-th ring that includes both of those edges due to Claim \ref{clm:ring_of_triangles_opt}, as those edges have different parity (see Figures \ref{fig:ring_of_triangles} and \ref{fig:fixed graph rings connections}). This contradicts that each ring of triangles contributes exactly 10 edges to $S$ (which is the optimal solution to cover the $\ell-th$ ring as well). Thus, $i=i'$. A similar argument gives that $j=j'$, completing the proof.
\end{proof}

\textbf{Adding edges corresponding the strings $x$ and $y$.}
Given two strings $x, y \in \{0,1\}^{{\left(\mu_{k+1}+1\right)}^2}$, we modify the graph $G$ accordingly to obtain a graph $G_{x,y}$.
Assuming that the strings are indexed by a pair of two indices of the form $\left(i,j\right), 0\le i,j \le \mu_{k+1}$, we update the graph $G$ by adding the edges $\{{a_1}^i, {a_2}^j\}$ if $x_{ij}=0$ and the edges $\{{b_1}^i, {b_2}^j\}$ if $y_{ij}=0$.

Before proving the next lemma, implying equivalence between Set Disjointness and MTET in our lower bound graphs, we first provide some intuition. We partition the nodes to two parts, one on Alice's side, informally all sets of nodes with $A_1$ or $A_2$ in their name, and the other nodes on Bob's side (see figure \ref{fig:lower bound fixed graphs}). The structure of the graph forces Alice to choose $\lceil\log\left(\mu_{k+1}+1\right)\rceil$ bit-edges, and thus forces Bob to choose the same mirrored bit-edges on his side due to the 20-ring of triangles connection between those bit-edges (see Figure \ref{fig:fixed graph rings connections}). The chosen bit-edges correspond to the binary representations of the nodes $a_1^i, a_2^j, b_1^i, b_2^j$, therefore omitting the unnecessary edges $\{a, a_1^i\}, \{a, a_2^j\}, \{b, b_1^i\}, \{b, b_2^j\}$ from the MTET iff the edges $\{a_1^i,a_2^j\}, \{b_1^i, b_2^j\}$ don't exist, which implies the disjointness of the inputs. If Alice and Bob chose bit edges with corresponding binary representations that do not match, then the solution size will increase.
	The connections to the other cliques $C_{A_1}, C_{A_2}, C_{B_1}, C_{B_2}$ (See Figure \ref{fig:fixed graph connectors connections}) are necessary to prevent other bit edges from being chosen (they correspond to exactly one binary representation).

\begin{lemma}\label{lem:lower_bound_graph_disj}
	The graph $G_{x,y}$ has a triangle edges-transversal of cardinality $M=4\tau\left(K_{k+1}\right)+4\left(\mu_{k+1}+1\right)+2\cdot10\lceil\log\left(\mu_{k+1}+1\right)\rceil$ iff DISJ$\left(x,y\right)=\texttt{false}$.
\end{lemma}

\begin{proof}
	If DISJ$\left(x,y\right)$=\texttt{false}, then there exists a pair of indices $\left(i,j\right)\in{\{0,\dots,\mu_{k+1}\}}^2$ such that $x_{ij}=y_{ij}=1$. Consider the following sets of edges:
	\begin{itemize}
		\item A subset of the edges between the center node $a$ and the clique $C_{A_1}$:\\ $E_{a,C_{A_1}}=\{\{a, c_{A_1}^\ell\}: 0\le\ell\le\mu_{k+1}, \ell\ne i\}$.
		\item A subset of the edges between the center node $a$ and the clique $C_{A_2}$:\\ $E_{a,C_{A_2}}=\{\{a, c_{A_2}^\ell\}: 0\le\ell\le\mu_{k+1}, \ell\ne j\}$.
		\item A subset of the edges between the center node $b$ and the clique $C_{B_1}$:\\ $E_{b,C_{B_1}}=\{\{b, c_{B_1}^\ell\}: 0\le\ell\le\mu_{k+1}, \ell\ne i\}$.
		\item A subset of the edges between the center node $b$ and the clique $C_{B_2}$:\\ $E_{b,C_{B_2}}=\{\{b, c_{B_2}^\ell\}: 0\le\ell\le\mu_{k+1}, \ell\ne j\}$.
	\end{itemize}
	Let $S_{C_{A_1}}, S_{C_{A_2}}, S_{C_{B_1}}, S_{C_{B_2}}$ be the optimal triangle edge transversals of each of the cliques $C_{A_1}\cup \{a\}, C_{A_2}\cup \{a\}, C_{B_1}\cup \{b\}, C_{B_2}\cup \{b\}$ respectively, such that $E_{a,C_{A_1}} \subseteq S_{C_{A_1}}, E_{a,C_{A_2}} \subseteq S_{C_{A_2}}, E_{b,C_{B_1}} \subseteq S_{C_{B_1}}, E_{b,C_{B_2}} \subseteq S_{C_{B_2}}$, which exist due to Corollary \ref{col:complete_graph_edge_transversal}. 
	
	Let us denote the union of these triangle edge transversals by $S_{cliques}  = S_{C_{A_1}} \cup S_{C_{A_2}} \cup S_{C_{B_1}} \cup S_{C_{B_2}}$.  Consider the optimal solutions of each ring of triangles (Figure \ref{fig:fixed graph rings connections}) that include the edges $\{\{a, u\}, u\in \bin(a_1^i)\}$ and the edges $\{\{a, v\}, v\in \bin(a_2^j)\}$ (these solutions also include the edges $\{\{b, u\}, u\in \bin(b_1^i)\}$, $\{\{b, v\}, v\in \bin(b_2^j)\}$). Let us denote the union of the solutions of each ring of triangles above by $S_{rings}$. Let us denote $S_H = \{\{a, h_{A_1}^i\},  \{a, h_{A_2}^j\},  \{b, h_{B_1}^i\},  \{b, h_{B_2}^j\}\}$ and $S_{bits} =\{\{a,a_1^\ell\}: 0\le\ell\le\mu_{k+1}, \ell\ne\ i\} \cup \{\{a,a_2^\ell\}: 0\le\ell\le\mu_{k+1}, \ell\ne\ j\} \cup \{\{b,b_1^\ell\}: 0\le\ell\le\mu_{k+1}, \ell\ne\ i\} \cup \{\{b,b_2^\ell\}: 0\le\ell\le\mu_{k+1}, \ell\ne\ j\}$.
	Thus, we claim that the subset $S = S_{cliques} \cup S_{rings} \cup S_{bits} \cup S_H$  is a triangle edge transversal of $G_{x,y}$ of the required cardinality.
	
	It is easy to validate that the cardinality of $S$ is equal to the required cardinality in the statement of the lemma. We must show that $S$ is a triangle edge transversal. The triangles in each clique $C_R$, for $R\in\{A_1,A_2,B_1,B_2\}$, are covered by $S_{C_R}$. The triangles between the bit-nodes ($A_1,A_2,B_1,B_2$), the center node ($a,b$), the connector nodes $H_R$, for $R\in\{A_1,A_2,B_1,B_2\}$, and the cliques, are covered by either $S_{cliques}$, $S_H$ or $S_{bits}$ (see the triangles in Figure \ref{fig:fixed graph connectors connections}). The triangles between the bit nodes, the bit-gadgets and the center nodes are covered by either $S_{rings}$ (specifically, for the triangles that contain the edge
	$\{a,a_{A_1}^i\}, \{a,a_{A_2}^j\}, \{b,b_{B_1}^i\}$ or $\{b,b_{B_2}^j\}$) or $S_{bits}$.
	The triangles in the 20-rings of triangles are covered by $S_{rings}$.
	Lastly, all the other triangles in the graph $G_{x,y}$ are created after adding the edges corresponding to the strings $x,y$, and they are of the form $\{a,a_{A_1}^{i'},a_{A_2}^{j'}\}$ or $\{b,b_{B_1}^{i'},b_{B_2}^{j'}\}$. Those triangles are covered by $S_{bits}$, which includes at least one edge from every possible triplet of that form except for the two triplets $\{a,a_{A_1}^{i},a_{A_2}^{j}\}$ and $\{b,b_{B_1}^{i},b_{B_2}^{j}\}$. However, these triplets are not triangles in the graph, as it is given that $x_{ij} = 1$ (i.e., the edges $\{a_{A_1}^i, a_{A_2}^j\}, \{b_{B_1}^i, b_{B_2}^j\}$ do not exist).
	
	On the other hand, if the triangle edge transversal $S$ is of the given size, then by Claim \ref{clm:bit_correspondence}, there exist two indices $i,j$ such that none of the edges $\{a,a_{A_1}^i\}, \{a,a_{A_2}^j\}, \{b,b_{B_1}^i\}, \{b,b_{B_2}^j\}$ are in $S$. Thus, the edges $\{a_{A_1}^i,a_{A_2}^j\}, \{b_{B_1}^i,b_{B_2}^j\}$ do not exist, as otherwise they must be in $S$ and the size of the cover is larger than given.
	The non-existence of the edges $\{a_{A_1}^i,a_{A_2}^j\}, \{b_{B_1}^i,b_{B_2}^j\}$ corresponds to the strings $x,y$ having a bit 1 in the index $\left(i,j\right)$. Thus, the inputs are not disjoint.
\end{proof}

We are now finally ready to prove our main theorem.

\begin{proof}[Proof of Theorem \ref{theorem:MTET_lower_bound_congest}]
	Let $MA_i^\ell\subseteq M_i^\ell$ be the ring-auxiliary nodes $MA_i^\ell = \{m_{i,\ell}^j\mid j\in\{0,1,2,10,11,12,13\}\}$ and let $MB_i^\ell\subseteq M_i^\ell$ be the ring-auxiliary nodes $MB_i^\ell = \{m_{i,\ell}^j\mid j\in\{3,4,5,6,7,8,9\}\}$ for $i\in{1,2}$ and $0\le \ell\le \lceil\log\left(\mu_{k+1}+1\right)\rceil-1$. Let $MA$ be the union of the sets $MA_i^\ell$ and let $MB$ be the union of the sets $MB_i^\ell$.
	We divide the nodes of the graph $G$ (and $G_{x,y}$) into two sets. One set is $V_A=A_1\cup A_2\cup F_{A_1}\cup F_{A_2}\cup H_{A_1}\cup H_{A_2}\cup C_{A_1}\cup C_{A_2}\cup MA$ and the other set is $V_B=B_1\cup B_2\cup F_{B_1}\cup F_{B_2}\cup H_{B_1}\cup H_{B_2}\cup C_{B_1}\cup C_{B_2}\cup MB$. 
	
	Note that the number of nodes $n$ is $\Theta\left(k\right)$, as Claim \ref{clm:max_edges_bounds} gives that $\mu_{k+1}=\Theta\left(k\right)$. In particular,  $|A_1|=|A_2|=|H_{A_1}|=|H_{A_2}|=\Theta\left(k\right)$. In addition, the size of the inputs $x$ and $y$, which is equal to $\mu_{k+1}^2$, is $K=\Theta\left(k^2\right)=\Theta\left(n^2\right)$. Furthermore, the cut between the two parts $V_A,V_B$ consists only of the edges of the rings of triangles. Precisely, each ring contributes a constant number of edges to the cut. Hence, the cut size is $\Theta\left(\log\mu_{k+1}\right)=\Theta\left(\log n\right)$. Since Lemma \ref{lem:lower_bound_graph_disj} proves that $G_{x,y}$ is a family of lower bound graphs for triangle edge transversal, and since the communication complexity of set disjointness is linear in the input size, then by applying Theorem \ref{thm:CC_reduction} on the partition $\left\{V_A,V_B\right\}$, we deduce that any algorithm in the \congest model for deciding whether a given graph has a triangle edge transversal of size $M=4\tau\left(K_{k+1}\right)+4\left(\mu_{k+1}+1\right)+2\cdot10\lceil\log\left(\mu_{k+1}+1\right)\rceil$ requires at least a near quadratic number of rounds $\Omega\left(K/\log^2\left(n\right)\right)=\Omega\left(n^2/\log^2\left(n\right)\right)$. 
\end{proof}

\vspace{-0.5cm}
\section{A $(1+\epsilon)$-Approximation for MTET in the \local Model}
\vspace{-0.2cm}
Despite MTET being a global problem, we use the approach of~\cite{GhaffariKM17} to show that it can be well approximated very efficiently.

\trianglesLOCALonePlusEps*

\begin{proof}[Proof of Theorem \ref{theorem:OnePlusEpsInLOCAL}]
	The proof is essentially the ball-carving technique that appears in~\cite{GhaffariKM17} for the problem of approximating minimum dominating set, which we state here for our problem of Triangle Edge Transversal. The rough high level overview is that we take a node and start growing a ball around it, and with each increase in the radius the center node collects all edges in the ball and computes an optimal solution, until the ratio between the new and previous optimal solutions is capped by $(1+\epsilon)$. At this time, we carve out the inner ball and pay for the outer cover, but the overhead of our solution is only $(1+\epsilon)$. Because there are only $m=O(n^2)$ edges, this process can only last for at most $O(\log(m))=O(\log(n))$ iterations in the worst case. We then repeat the ball-carving process with another center node. To obtain parallelism in this process, it is run over a network decomposition, which allows working on multiple center nodes concurrently.
	
	Formally, for a node $v$ and a radius $r$, we denote by $B_{r}(v)$ the ball of radius $r$ around $v$, i.e., the subgraph induced by the nodes that are distant at most $r$ from $v$. We further denote by $g(v,r)$ the size of an optimal edge cover of triangles that have at least one edge in $B_r(v)$ (note that such triangles, and hence also edges in the cover, may be outside of $B_r(v)$, but edges in the cover must have an endpoint in $B_r(v)$).
	
	Consider the following algorithmic template. Let $v_1,\dots,v_n$ be an arbitrary order of the nodes in the graph $G$. We process the nodes in iterations according to this order. Let $R=poly(\log{n},1/\epsilon)$. For $1\leq i\leq n$, an iteration $i$ is composed of steps, where for $1\leq j \leq R$, in step $j$ we compute an optimal triangle cover of $B_{r_j}(v_i)$, where $r_j=2j-1$. If its size $g(v_i,r_j)$ is at most $(1+\epsilon)\cdot g(v_i,r_{j-1})$, then we mark its edges as $C_i$, we set $r(v_i)=r_j$ and we finish this iteration and remove all edges in $B_{r(v_i)}(v_i)=B_{r_{j}}(v_i)$ from the graph (we stress that we keep the edges that have at least one endpoint outside $B_{r_j}(v_i)$). Note that since the values of $g(v_i,r_j)$ grow by a factor of at least $(1+\epsilon)$ for every increase in $j$, we have that $r(v_i)$ is bounded by some value in $poly(\log{n},1/\epsilon)$, which we set as $R$.
	
	By its definition, $\bigcup_{1\leq i\leq n} C_i$ is a triangle cover of the graph $G$: given any triangle $t$ in the graph, consider the first $i$ such that some edge of $t$ belongs to $B_{r(v_i)}(v_i)$. Then $C_i$ includes an edge of $t$. There must exist such $i$ or else the process is not yet finished. We further claim that $\bigcup_{1\leq i\leq n} C_i$ is a $(1+\epsilon)$ approximation to an optimal solution $OPT$. For $1\leq i\leq n$, let $D_i$ be the edges of $OPT$ that cover $B_{r(v_i)-2}(v_i)$. Notice that $|C_i| \leq g(v_i,r(v_i)) \le (1+\epsilon)g(v_i,r(v_i)-2) \leq (1+\epsilon)|D_i|$. Further, the sets $D_i$ are disjoint. Thus, we have that the size of our solution is at most $|\bigcup_{1\leq i\leq n} C_i| \leq \bigcup_{1\leq i\leq n} |C_i| \leq \bigcup_{1\leq i\leq n} (1+\epsilon)|D_i| \leq (1+\epsilon)OPT$.
	
	It remains to show how to implement the above template in the \local model within $poly(\log{n},1/\epsilon)$ rounds. Consider the power graph $G^R$, which has the same set of nodes as $G$ and in which two nodes are connected by an edge if they are within distance $R$ in $G$. The nodes in $G$ invoke a $(c,d)$-network decomposition algorithm over $G^R$, where a network decomposition partitions the graph into disjoint clusters of diameter at most $d$, with a $c$-coloring for the cluster graph. We order the nodes according to the order of $(color_v,id_v)$, where $color_v$ is the color of the cluster to which $v$ belongs and $id_v$ is its identifier. The algorithm proceeds in $c$ phases, where in phase $1\leq\ell\leq c$ each node in a cluster of color $\ell$ collects all edge information from all nodes in its cluster and simulates the above template locally. Because any two nodes within distance $R$ are either in the same cluster or in neighboring clusters and thus with different colors, this approach correctly simulates the template.
	
	The round complexity of the algorithm depends on three parameters -- the time $T$ it takes to construct the network decomposition, the time $R$ it takes to simulate the construction over $G^R$, the number of colors $c$, and the time $d$ it takes to collect cluster information. We have that $R=poly(\log{n},1/\epsilon)$, and it is known that network decompositions with $c,d=O(\log{n})$ can be found in $poly(\log{n})$ rounds, due to the celebrated randomized algorithm of~\cite{LinialS93} or the recent deterministic breakthrough algorithm of~\cite{RozhonG20}. We thus get the claimed number of rounds, which completes the proof.
\end{proof}

\vspace{-0.3cm}
\section{Faster Approximations for the Minimum Triangle Edge Transversal}
\label{sec:trianglesFasterApprox}
\vspace{-0.2cm}
In this section, we provide a reduction from the MTET problem to the MHVC problem and use it to show how the $(3+\epsilon)$-approximation algorithm for MHVC introduced in \cite[Section 3]{BenBasatEKS19} can be simulated in the \local and \congest models to solve MTET. With a slight increase in the time complexity, the algorithm can also approximate MHVC by a factor of $3$. The reduction and the adjusted algorithm are also applicable for the weighted case.

The faster approximation we obtain for MTET in the \local and \congest models is summarized in the following.

\trianglesLOCALthreeAndthreePlusEps*

We denote the input graph for the MTET problem by $G=(V,E)$. An MHVC algorithm considers a communication network on a hypergraph $H$, where every node $v_H$ and hyperedge $e_H$ have their own computation units, and $v_H$ can communicate with $e_H$ if and only if $v_H\in e_H$, i.e., every hyperedge can communicate with its nodes and every node can communicate with all the hyperedges it is a part of. The complexity in this model is measured by rounds. 
The degree of each node is bounded by $\Delta$ and the cardinality of every hyperedge is bounded by $f$ (the rank of the hypergraph).

~\\\textbf{Reduction of MTET to MHVC.} We construct the hypergraph $H_G=(V_{H_G},E_{H_G})$ as the following. Each edge in the original graph $G$ becomes a node in $H_G$, i.e., $V_{H_G} = E$. Every triangle $t=\{e_1,e_2,e_3\}$ in the original graph becomes a hyperedge in $H_G$, i.e., $E_{H_G}=\{\{e_1,e_2,e_3\}: e_1,e_2,e_3\in E \text{ form a triangle}\}$. For the weighted case, the weight of every node in $H_G$ is the same as the weight of the corresponding edge in $E$.

We call the obtained hypergraph \emph{the reduced hypergraph of $G$} and denote it by $H_G$. It is easy to see that the reduced hypergraph $H_G$ is $3$-uniform. Specifically, the rank of the graph is $f=3$. We denote the node in the reduced hypergraph $H_G$ corresponding to an edge $e$ in $G$ by $v^e_{H_G}$, and the hyperedge corresponding to a triangle $t$ in $G$ by $e^t_{H_G}$. We also denote the set of nodes of the edge $e$ and the triangle $t$ by $V_e$ and $V_t$, respectively.

In our case, each node in the reduced hypergraph $v^e_{H_G}$ corresponds to the edge $e=\{u,v\}$ in the original graph. Therefore, the messages of this node will be simulated by the two endpoints of the corresponding edge, i.e., $V_e$. The hyperedges corresponding to the triangles in the original graph will be simulated by the three nodes that forms the triangle, i.e., $V_t$.

We claim that any algorithm for MHVC that runs on the reduced hypegraph in \local can be simulated with the same complexity in the original graph. 
We further claim that any algorithm for MHVC that runs on $H_G$ in \congest can be simulated with a multiplicative overhead of $\Delta$ over the original complexity.

\begin{restatable}{myclaim}{AlgorithmReductionComplexity}
	\label{clm:AlgorithmReductionComplexity}
	Let $G=(V,E)$ be an undirected graph and $H_G$ be the reduced hypergraph of $G$. Then:
	\begin{enumerate}
		\item Any algorithm that solves MHVC in $H_G$ in the \local model that runs in $O(r)$ rounds can be simulated in the graph $G$ in $O(r)$ rounds.
		\item Any algorithm that solves MHVC in $H_G$ in the \congest model that runs in $O(r)$ rounds can be simulated in the graph $G$ in $O(\Delta\cdot r)$ rounds.
	\end{enumerate}
\end{restatable}

To prove Claim \ref{clm:AlgorithmReductionComplexity}, we first show that any round of a MHVC algorithm on the reduced hypergraph can be simulated in $O(1)$ and $O(\Delta)$ rounds in the \local and the \congest model respectively, under the assumption that the information in the endpoints that simulate the same node/hyperedge is \emph{coherent}, as defined shortly. We then show how to preserve the coherency between the simulating endpoints for the next iteration. 

\begin{definition}
	Let $H=(V_H,E_H)$ be a hypergraph and let $Alg$ be a distributed algorithm on hypergraphs. Then the \emph{essential information} that a node knows after round $i$, denoted $I^{Alg}_{v_H, i}$, or a hyperedge knows, denoted $I^{Alg}_{e_H, i}$, is the set of messages receive/sent and variables' states in all rounds in $Alg$ up to round $i$. For essential information $I$, we say that a set of nodes are \emph{$I$-coherent} if all the nodes in that set know $I$.
\end{definition}

The following is a simple observation follows directly from the definition of $Alg$.
\begin{observation}\label{obs:EssentialInformation}
	The essential information of node $v_H$ or a hyperedge $e_H$ after round $i$ is sufficient for computing the messages that the node or the hyperedge need to send in round $i+1$ in $Alg$.
\end{observation}

\begin{myclaim}\label{clm:AlgorithmReductionSimulation}
	Let $G=(V,E)$ be an undirected graph, let $H_G$ be the reduced hypergraph of $G$, let $Alg$ be a MHVC algorithm, and let $I^{Alg}_{v_{H_G}, i}$ and $I^{Alg}_{e_{H_G},i}$ be the essential information that a node $v_{H_G}$ and hyperedge $e_{H_G}$ in $H_G$ know after round $i$, respectively. For any edge $e$ and triangle $t$ in $G$, assume that the nodes $V_e$ are $I^{Alg}_{v^e_{H_G}, i}$-coherent and the nodes in $V_t$  are $I^{Alg}_{e^t_{H_G}, i}$-coherent. Then the nodes of $G$ can simulate round $i+1$ of $Alg$ in the following sense.
	\begin{enumerate}
		\item Each node $v$ in $G$ can compute the set of triangles that it is an endpoint of in $O(1)$ and $O(\Delta)$ rounds in the \local and the \congest models,  respectively.
		\item Any messages sent from a node in $H_G$ to its incident hyperedges in $Alg$ in round $i+1$ can be simulated in $O(1)$ and $O(\Delta)$ rounds in the \local and the \congest models, respectively.
		\item Any messages sent from a hyperedge in $H_G$ to its incident nodes in $Alg$ in round $i+1$ can be simulated in 0 rounds in the \local and the \congest models, respectively.
		\item  After simulating the process of sending the messages in the above two items, for any edge $e$ and triangle $t$ in $G$, the nodes of $V_e$ are $I^{Alg}_{v^e_{H_G}, i+1}$-coherent and the nodes of $V_t$ are  $I^{Alg}_{e^t_{H_G}, i+1}$-coherent.
	\end{enumerate}
\end{myclaim}

\begin{proof}
	We prove the four items of the simulation:
	\begin{enumerate}
		\item This is simply done by sending the list of neighbors to the each neighbor. Therefore, each node can check if two of its neighbor are also neighbors themselves and if so then add the triangle to its list.
		\item Note that each of the hyperedges that are incident to the node $v^e_{H_G}$ receive the corresponding message. By Observation \ref{obs:EssentialInformation}, the messages that are set to be sent by the node can be computed as a function of the essential information $I^{Alg}_{v^e_{H_G}, i}$. Thus, both of the simulating nodes of $e$, i.e, $V_e$, can compute the messages that the corresponding node $v^e_{H_G}$ needs to send, since they are $I^{Alg}_{v^e_{H_G}, i}$-coherent. The nodes of $V_e$ must inform $V_t$ of the messages that the node $v^e_{H_G}$ would send to $e^t_{H_G}$ for every triangle $t$ that $e$ participates in. Therefore, the nodes must send the corresponding message to the corresponding third node in $V_t$ that is not in $e$. This can be done in $O(1)$ and $O(\Delta)$ rounds in the \local and \congest models,  respectively: for the \local model this is immediate, and for the \congest model, note that each edge can participate in at most $\Delta-2$ triangles, thus, the number of messages sent on an edge is at most $2\Delta-4$.
		\item Note that $V_e\subseteq V_t$ for an edge $e$ in triangle $t$, thus the nodes the edges of the triangle $t$ already know the messages that the triangle (corresponding to the hyperedge $e^t_{H_G}$) should send to its edges ($v^e_{H_G}$) and no rounds are needed for simulating.
		\item We know that the messages that each node or hyperedge sent or received in round $i+1$ are known to the simulating nodes $V_e$ or $V_t$, respectively. Therefore, we only need to show that the nodes can compute the new state of the variables it contains for an edge or a triangle that contains it. Note that the state of the variables are computed as a function of the previous states, the previous messages and the new messages. The previous states and messages are known to every nodes, as a part of the definition of coherency, and the new messages have been computed in $O(1)$ and $O(\Delta)$ rounds in the \local and \congest models, respectively. Therefore, the nodes of $V_e$ and $V_t$ are $I^{Alg}_{v^e_{H_G}, i+1}$-coherent and $I^{Alg}_{e^t_{H_G}, i+1}$-coherent, respectively.
	\end{enumerate}
\end{proof}

Now we can prove Claim \ref{clm:AlgorithmReductionComplexity}, essentially by an induction over Claim~\ref{clm:AlgorithmReductionSimulation}.

\begin{proof}[Proof of Claim \ref{clm:AlgorithmReductionComplexity}]
	We show by induction on the coherency of round $i$ that each round of a MHVC algorithm can be simulated in $O(1)$ and $O(\Delta)$ rounds in the \local and \congest models, respectively, and thus the round complexity of the algorithm is multiplied by these factors. Note that showing coherency every round is sufficient for proving that the rounds of the algorithm can be simulated in the complexities stated above due to Items (2) and (3) in Claim \ref{clm:AlgorithmReductionSimulation}.
	
	For the base case ($i=0$), Item (1) in Claim \ref{clm:AlgorithmReductionSimulation} obtains coherency for the start of the simulation, that is, for every edge $e$ and triangle $t$ in $G$, we have that the nodes in $V_e$ and $V_t$ are $I^{Alg}_{v^e_{H_G}, i}$-coherent and $I^{Alg}_{e^t_{H_G}, i}$-coherent, respectively. This means that they hold the inputs of the nodes in $H_G$ which they need to simulate.

	Assuming that the simulating nodes $V_e$ or $V_t$ are coherent (specifically, $V_e$ and $V_t$ are $I^{Alg}_{v^e_{H_G}, i}$-coherent and $I^{Alg}_{e^t_{H_G}, i}$-coherent, respectively), Item (4) in Claim \ref{clm:AlgorithmReductionSimulation} states that the nodes can run a protocol to stay coherent after round $i+1$ in $O(1)$ and $O(\Delta)$ rounds in the \local and \congest models, respectively. Therefore, round $i+1$ can be simulated in the same complexities. 
\end{proof}

We now plug the MHVC algorithm of \cite[Section 3]{BenBasatEKS19} into our reduction. The following claim summarizes its round complexity.

\begin{myclaim}[{\cite[Corollaries 4.10, 4.12 and Appendix B]{BenBasatEKS19}}]\label{clm:MWHVCAlgorithmComplexity} 
	There are distributed algorithms for MHVC that:
	\begin{enumerate}
		\item compute an $f$-approximation in $O(f\log n)$ rounds in the \local model,
		\item  for $f=O(1)$ and $\epsilon=2^{-O\left((\log \Delta)^{0.99}\right)}$, compute an $(f+\epsilon)$-approximation in $O\left(\frac{\log\Delta}{\log\log\Delta}\right)$ rounds in the \local model, and 
		\item can be adapted to the \congest model without affecting the round complexity.
	\end{enumerate}
\end{myclaim}

\begin{proof}
	[Proof of Theorem \ref{theorem:threeInLOCAL}]
	The proof combines Claims \ref{clm:AlgorithmReductionComplexity} and \ref{clm:MWHVCAlgorithmComplexity}. By Claim \ref{clm:AlgorithmReductionComplexity}, the execution of the MHVC algorithm of Claim \ref{clm:MWHVCAlgorithmComplexity} on the reduced hypergraph can be simulated on the original graph in the \local model in the same round complexity and in the \congest model up to a factor $\Delta$ of the round complexity on the hypergraph.
	Then, for its solution to MTET, each node $v$ marks the edges in $\{e=\{u,v\}\mid u\in N(v)\}$ that correspond to the nodes in the solution for MHVC (of $H_G$), which it knows about since it simulates the nodes $v^e_{H_G}$. 
	
	This yields a $3$-approximation algorithm that runs in $O(f\log |V_{H_G}|)=O(3\log |E_G|)=O(\log n^2)=O(\log n)$ and $O(\Delta\cdot f\log |V_{H_G}|)=O(\Delta\log n)$ rounds in the \local and \congest models, respectively. For the $(3+\epsilon)$-approximation, we obtain round complexities of $O(\frac{\log \Delta_H}{\log\log\Delta_H})=O(\frac{\log n}{\log\log n})$ and $O(\frac{\Delta\log \Delta_H}{\log\log\Delta_H})=O(\Delta\frac{\log n}{\log\log n})$, respectively. The correctness of the algorithm (feasibility and approximation ratio) are directly derived from the correctness of the MHVC algorithm (see \cite[Section 4.1]{BenBasatEKS19}).
\end{proof}

~\\
\noindent\textbf{Acknowledgements.} 
This project has received funding from the European Union’s Horizon 2020 research and innovation programme under grant agreement no. 755839. The authors thank Seri Khoury for many useful discussions.

\bibliographystyle{splncs04}
\bibliography{bibliography}

\begin{thebibliography}{10}
\providecommand{\url}[1]{\texttt{#1}}
\providecommand{\urlprefix}{URL }
\providecommand{\doi}[1]{https://doi.org/#1}

\bibitem{AbboudCK16}
Abboud, A., Censor{-}Hillel, K., Khoury, S.: Near-linear lower bounds for
  distributed distance computations, even in sparse networks. In: Gavoille, C.,
  Ilcinkas, D. (eds.) Distributed Computing - 30th International Symposium,
  {DISC} 2016, Paris, France, September 27-29, 2016. Proceedings. Lecture Notes
  in Computer Science, vol.~9888, pp. 29--42. Springer (2016).
  \doi{10.1007/978-3-662-53426-7\_3},
  \url{https://doi.org/10.1007/978-3-662-53426-7\_3}

\bibitem{AbboudCHKL}
Abboud, A., Censor{-}Hillel, K., Khoury, S., Lenzen, C.: Fooling views: a new
  lower bound technique for distributed computations under congestion.
  Distributed Computing  \textbf{33},  545--–559 (2020).
  \doi{10.1007/s00446-020-00373-4},
  \url{https://doi.org/10.1007/s00446-020-00373-4}

\bibitem{AbboudCKP21}
Abboud, A., Censor{-}Hillel, K., Khoury, S., Paz, A.: Smaller cuts, higher
  lower bounds. {ACM} Trans. Algorithms  \textbf{17}(4),  30:1--30:40 (2021).
  \doi{10.1145/3469834}, \url{https://doi.org/10.1145/3469834}

\bibitem{bacrach2019hardness}
Bacrach, N., Censor-Hillel, K., Dory, M., Efron, Y., Leitersdorf, D., Paz, A.:
  Hardness of distributed optimization. In: Proceedings of the 2019 ACM
  Symposium on Principles of Distributed Computing. pp. 238--247 (2019)

\bibitem{Bar-YossefJKS04}
Bar{-}Yossef, Z., Jayram, T.S., Kumar, R., Sivakumar, D.: An information
  statistics approach to data stream and communication complexity. J. Comput.
  Syst. Sci.  \textbf{68}(4),  702--732 (2004).
  \doi{10.1016/j.jcss.2003.11.006},
  \url{https://doi.org/10.1016/j.jcss.2003.11.006}

\bibitem{BenBasatEKS19}
Ben{-}Basat, R., Even, G., Kawarabayashi, K., Schwartzman, G.: Optimal
  distributed covering algorithms. In: Robinson, P., Ellen, F. (eds.)
  Proceedings of the 2019 {ACM} Symposium on Principles of Distributed
  Computing, {PODC} 2019, Toronto, ON, Canada, July 29 - August 2, 2019. pp.
  104--106. {ACM} (2019). \doi{10.1145/3293611.3331577},
  \url{https://doi.org/10.1145/3293611.3331577}

\bibitem{BrakerskiP11}
Brakerski, Z., Patt{-}Shamir, B.: Distributed discovery of large near-cliques.
  Distributed Comput.  \textbf{24}(2),  79--89 (2011).
  \doi{10.1007/s00446-011-0132-x},
  \url{https://doi.org/10.1007/s00446-011-0132-x}

\bibitem{CH22}
Censor{-}Hillel, K.: Distributed subgraph finding: Progress and challenges.
  CoRR  \textbf{abs/2203.06597} (2022). \doi{10.48550/arXiv.2203.06597},
  \url{https://doi.org/10.48550/arXiv.2203.06597}

\bibitem{CFSV19}
Censor{-}Hillel, K., Fischer, E., Schwartzman, G., Vasudev, Y.: Fast
  distributed algorithms for testing graph properties. Distributed Comput.
  \textbf{32}(1),  41--57 (2019). \doi{10.1007/S00446-018-0324-8},
  \url{https://doi.org/10.1007/s00446-018-0324-8}

\bibitem{CensorFGLLO21}
Censor{-}Hillel, K., Fischer, O., Gonen, T., Gall, F.L., Leitersdorf, D.,
  Oshman, R.: Fast distributed algorithms for girth, cycles and small
  subgraphs. CoRR  \textbf{abs/2101.07590} (2021),
  \url{https://arxiv.org/abs/2101.07590}

\bibitem{Censor-HillelKK19}
Censor{-}Hillel, K., Kaski, P., Korhonen, J.H., Lenzen, C., Paz, A., Suomela,
  J.: Algebraic methods in the congested clique. Distributed Comput.
  \textbf{32}(6),  461--478 (2019). \doi{10.1007/s00446-016-0270-2},
  \url{https://doi.org/10.1007/s00446-016-0270-2}

\bibitem{Censor-HillelLT20}
Censor{-}Hillel, K., Leitersdorf, D., Turner, E.: Sparse matrix multiplication
  and triangle listing in the congested clique model. Theor. Comput. Sci.
  \textbf{809},  45--60 (2020). \doi{10.1016/j.tcs.2019.11.006},
  \url{https://doi.org/10.1016/j.tcs.2019.11.006}

\bibitem{CensorHLV22}
Censor{-}Hillel, K., Leitersdorf, D., Vulakh, D.: Deterministic near-optimal
  distributed listing of cliques. In: Proc. of the Symp. on Principles of
  Distributed Comp. (PODC). pp. 271--280 (2022)

\bibitem{ChangPSZ21}
Chang, Y.J., Pettie, S., Saranurak, T., Zhang, H.: Near-optimal distributed
  triangle enumeration via expander decompositions. Journal of the {ACM}
  \textbf{68}(3),  1--36 (2021)

\bibitem{ChangS20}
Chang, Y., Saranurak, T.: Deterministic distributed expander decomposition and
  routing with applications in distributed derandomization. In: Proceedings of
  the 61st {IEEE} Annual Symposium on Foundations of Computer Science (FOCS).
  pp. 377--388 (2020). \doi{10.1109/FOCS46700.2020.00043},
  \url{https://doi.org/10.1109/FOCS46700.2020.00043}

\bibitem{DasSarmaHKKNPPW11}
Das~Sarma, A., Holzer, S., Kor, L., Korman, A., Nanongkai, D., Pandurangan, G.,
  Peleg, D., Wattenhofer, R.: Distributed verification and hardness of
  distributed approximation. In: Proceedings of the Forty-Third Annual ACM
  Symposium on Theory of Computing. p. 363–372. STOC '11, Association for
  Computing Machinery, New York, NY, USA (2011). \doi{10.1145/1993636.1993686},
  \url{https://doi.org/10.1145/1993636.1993686}

\bibitem{DolevLP12}
Dolev, D., Lenzen, C., Peled, S.: "tri, tri again": Finding triangles and small
  subgraphs in a distributed setting - (extended abstract). In: Aguilera, M.K.
  (ed.) Distributed Computing - 26th International Symposium, {DISC} 2012,
  Salvador, Brazil, October 16-18, 2012. Proceedings. Lecture Notes in Computer
  Science, vol.~7611, pp. 195--209. Springer (2012).
  \doi{10.1007/978-3-642-33651-5\_14},
  \url{https://doi.org/10.1007/978-3-642-33651-5\_14}

\bibitem{DruckerKO13}
Drucker, A., Kuhn, F., Oshman, R.: On the power of the congested clique model.
  In: Halld{\'{o}}rsson, M.M., Dolev, S. (eds.) {ACM} Symposium on Principles
  of Distributed Computing, {PODC} '14, Paris, France, July 15-18, 2014. pp.
  367--376. {ACM} (2014). \doi{10.1145/2611462.2611493},
  \url{https://doi.org/10.1145/2611462.2611493}

\bibitem{EdenFFKO22}
Eden, T., Fiat, N., Fischer, O., Kuhn, F., Oshman, R.: Sublinear-time
  distributed algorithms for detecting small cliques and even cycles.
  Distributed Comput.  \textbf{35}(3),  207--234 (2022).
  \doi{10.1007/s00446-021-00409-3},
  \url{https://doi.org/10.1007/s00446-021-00409-3}

\bibitem{ErdosGT96}
Erd{\"{o}}s, P., Gallai, T., Tuza, Z.: Covering and independence in triangle
  structures. Discret. Math.  \textbf{150}(1-3),  89--101 (1996).
  \doi{10.1016/0012-365X(95)00178-Y},
  \url{https://doi.org/10.1016/0012-365X(95)00178-Y}

\bibitem{EvenFFGLMMOORT17}
Even, G., Fischer, O., Fraigniaud, P., Gonen, T., Levi, R., Medina, M.,
  Montealegre, P., Olivetti, D., Oshman, R., Rapaport, I., Todinca, I.: Three
  notes on distributed property testing. In: Richa, A.W. (ed.) 31st
  International Symposium on Distributed Computing, {DISC} 2017, October 16-20,
  2017, Vienna, Austria. LIPIcs, vol.~91, pp. 15:1--15:30. Schloss Dagstuhl -
  Leibniz-Zentrum f{\"{u}}r Informatik (2017).
  \doi{10.4230/LIPIcs.DISC.2017.15},
  \url{https://doi.org/10.4230/LIPIcs.DISC.2017.15}

\bibitem{Fischer+SPAA18}
Fischer, O., Gonen, T., Kuhn, F., Oshman, R.: Possibilities and impossibilities
  for distributed subgraph detection. In: Proceedings of the 30th Symposium on
  Parallelism in Algorithms and Architectures ({SPAA}). pp. 153--162 (2018).
  \doi{10.1145/3210377.3210401}, \url{https://doi.org/10.1145/3210377.3210401}

\bibitem{FraigniaudO17}
Fraigniaud, P., Olivetti, D.: Distributed detection of cycles. In: Proceedings
  of the 29th {ACM} Symposium on Parallelism in Algorithms and Architectures
  (SPAA). pp. 153--162 (2017)

\bibitem{FrischknechtHW12}
Frischknecht, S., Holzer, S., Wattenhofer, R.: Networks cannot compute their
  diameter in sublinear time. In: Rabani, Y. (ed.) Proceedings of the
  Twenty-Third Annual {ACM-SIAM} Symposium on Discrete Algorithms, {SODA} 2012,
  Kyoto, Japan, January 17-19, 2012. pp. 1150--1162. {SIAM} (2012).
  \doi{10.1137/1.9781611973099.91},
  \url{https://doi.org/10.1137/1.9781611973099.91}

\bibitem{GhaffariKM17}
Ghaffari, M., Kuhn, F., Maus, Y.: On the complexity of local distributed graph
  problems. In: Hatami, H., McKenzie, P., King, V. (eds.) Proceedings of the
  49th Annual {ACM} {SIGACT} Symposium on Theory of Computing, {STOC} 2017,
  Montreal, QC, Canada, June 19-23, 2017. pp. 784--797. {ACM} (2017).
  \doi{10.1145/3055399.3055471}, \url{https://doi.org/10.1145/3055399.3055471}

\bibitem{GoldreichGR98}
Goldreich, O., Goldwasser, S., Ron, D.: Property testing and its connection to
  learning and approximation. J. {ACM}  \textbf{45}(4),  653--750 (1998).
  \doi{10.1145/285055.285060}, \url{https://doi.org/10.1145/285055.285060}

\bibitem{HirvonenRSS17}
Hirvonen, J., Rybicki, J., Schmid, S., Suomela, J.: Large cuts with local
  algorithms on triangle-free graphs. Electron. J. Comb.  \textbf{24}(4),
  P4.21 (2017),
  \url{http://www.combinatorics.org/ojs/index.php/eljc/article/view/v24i4p21}

\bibitem{HolzerP15}
Holzer, S., Pinsker, N.: Approximation of distances and shortest paths in the
  broadcast congest clique. In: Anceaume, E., Cachin, C., Potop{-}Butucaru,
  M.G. (eds.) 19th International Conference on Principles of Distributed
  Systems, {OPODIS} 2015, December 14-17, 2015, Rennes, France. LIPIcs,
  vol.~46, pp. 6:1--6:16. Schloss Dagstuhl - Leibniz-Zentrum f{\"{u}}r
  Informatik (2015). \doi{10.4230/LIPIcs.OPODIS.2015.6},
  \url{https://doi.org/10.4230/LIPIcs.OPODIS.2015.6}

\bibitem{HuangPZZ20}
Huang, D., Pettie, S., Zhang, Y., Zhang, Z.: The communication complexity of
  set intersection and multiple equality testing. In: Proceedings of the 2020
  {ACM-SIAM} Symposium on Discrete Algorithms (SODA). pp. 1715--1732 (2020).
  \doi{10.1137/1.9781611975994.105},
  \url{https://doi.org/10.1137/1.9781611975994.105}

\bibitem{IzumiG17}
Izumi, T., Gall, F.L.: Triangle finding and listing in {CONGEST} networks. In:
  Schiller, E.M., Schwarzmann, A.A. (eds.) Proceedings of the {ACM} Symposium
  on Principles of Distributed Computing, {PODC} 2017, Washington, DC, USA,
  July 25-27, 2017. pp. 381--389. {ACM} (2017). \doi{10.1145/3087801.3087811},
  \url{https://doi.org/10.1145/3087801.3087811}

\bibitem{KalyanasundaramS92}
Kalyanasundaram, B., Schnitger, G.: The probabilistic communication complexity
  of set intersection. {SIAM} J. Discret. Math.  \textbf{5}(4),  545--557
  (1992). \doi{10.1137/0405044}, \url{https://doi.org/10.1137/0405044}

\bibitem{KortsarzLN10}
Kortsarz, G., Langberg, M., Nutov, Z.: Approximating maximum subgraphs without
  short cycles. {SIAM} J. Discret. Math.  \textbf{24}(1),  255--269 (2010).
  \doi{10.1137/09074944X}, \url{https://doi.org/10.1137/09074944X}

\bibitem{Krivelevich95}
Krivelevich, M.: On a conjecture of tuza about packing and covering of
  triangles. Discret. Math.  \textbf{142}(1-3),  281--286 (1995).
  \doi{10.1016/0012-365X(93)00228-W},
  \url{https://doi.org/10.1016/0012-365X(93)00228-W}

\bibitem{KushilevitzN97}
Kushilevitz, E., Nisan, N.: Communication complexity. Cambridge University
  Press (1997)

\bibitem{Linial92}
Linial, N.: Locality in distributed graph algorithms. {SIAM} J. Comput.
  \textbf{21}(1),  193--201 (1992). \doi{10.1137/0221015},
  \url{https://doi.org/10.1137/0221015}

\bibitem{LinialS93}
Linial, N., Saks, M.E.: Low diameter graph decompositions. Comb.
  \textbf{13}(4),  441--454 (1993). \doi{10.1007/BF01303516},
  \url{https://doi.org/10.1007/BF01303516}

\bibitem{LotkerPPP05}
Lotker, Z., Patt{-}Shamir, B., Pavlov, E., Peleg, D.: Minimum-weight spanning
  tree construction in \emph{O}(log log \emph{n}) communication rounds. {SIAM}
  J. Comput.  \textbf{35}(1),  120--131 (2005).
  \doi{10.1137/S0097539704441848},
  \url{https://doi.org/10.1137/S0097539704441848}

\bibitem{PanduranganRS21}
Pandurangan, G., Robinson, P., Scquizzato, M.: On the distributed complexity of
  large-scale graph computations. {ACM} Trans. Parallel Comput.  \textbf{8}(2),
   7:1--7:28 (2021). \doi{10.1145/3460900},
  \url{https://doi.org/10.1145/3460900}

\bibitem{PelegBook}
Peleg, D.: Distributed Computing: A Locality-Sensitive Approach. Society for
  Industrial and Applied Mathematics, USA (2000)

\bibitem{PelegR00}
Peleg, D., Rubinovich, V.: A near-tight lower bound on the time complexity of
  distributed minimum-weight spanning tree construction. {SIAM} J. Comput.
  \textbf{30}(5),  1427--1442 (2000). \doi{10.1137/S0097539700369740},
  \url{https://doi.org/10.1137/S0097539700369740}

\bibitem{PettieS13}
Pettie, S., Su, H.: Fast distributed coloring algorithms for triangle-free
  graphs. In: Proceedings of the 40th International Colloquium on Automata,
  Languages, and Programming (ICALP). pp. 681--693 (2013)

\bibitem{Razborov90}
Razborov, A.A.: On the distributional complexity of disjontness. In:
  Proceedings of the 17th International Colloquium on Automata, Languages and
  Programming (ICALP). pp. 249--253 (1990)

\bibitem{RozhonG20}
Rozhon, V., Ghaffari, M.: Polylogarithmic-time deterministic network
  decomposition and distributed derandomization. In: Makarychev, K.,
  Makarychev, Y., Tulsiani, M., Kamath, G., Chuzhoy, J. (eds.) Proccedings of
  the 52nd Annual {ACM} {SIGACT} Symposium on Theory of Computing, {STOC} 2020,
  Chicago, IL, USA, June 22-26, 2020. pp. 350--363. {ACM} (2020).
  \doi{10.1145/3357713.3384298}, \url{https://doi.org/10.1145/3357713.3384298}

\bibitem{ruzsa1978triple}
Ruzsa, I.Z., Szemer{\'e}di, E.: Triple systems with no six points carrying
  three triangles. Combinatorics (Keszthely, 1976), Coll. Math. Soc. J. Bolyai
  \textbf{18}(939-945), ~2 (1978)

\end{thebibliography}

\end{document}